\documentclass[a4paper,10pt]{article}

\usepackage[english]{babel}
\usepackage{a4wide}
\usepackage[utf8]{inputenc}
\usepackage{amssymb,amsmath,amscd,amsfonts,amsthm,bbm,mathrsfs,enumerate}
\usepackage[shortlabels]{enumitem}
\usepackage{subfigure,float}
\usepackage{tikz,pgfplots}
\usepackage{color} 
\usepackage{yhmath}
\usepackage{accents}

\DeclareUnicodeCharacter{00A0}{~}

\theoremstyle{definition}
\newtheorem{theorem}{Theorem}
\newtheorem*{theorem*}{Statement}
\newtheorem{lemma}[theorem]{Lemma}
\newtheorem{corollary}[theorem]{Corollary}
\newtheorem{proposition}[theorem]{Proposition}

\newtheorem{remark}{Remark}
\newtheorem{example}{Example}

\newtheorem*{condition*}{Condition}
\newtheorem{assumption}{Assumption}

\DeclareMathOperator{\rank}{rank}
\DeclareMathOperator{\tr}{tr}

\DeclareMathOperator{\colspan}{span}
\DeclareMathOperator*{\aslim}{aslim}

\newcommand{\1}{\mathbbm 1}
\newcommand{\T}{{\mathsf T}} 

\newcommand{\ZZ}{\mathbb Z}
\newcommand{\CC}{\mathbb{C}}
\newcommand{\PP}{{{\mathbb P}}} 
\newcommand{\EE}{{{\mathbb E}}} 
\newcommand{\NN}{{{\mathbb N}}} 
\newcommand{\RR}{{{\mathbb R}}} 

\newcommand{\mcB}{{\mathscr B}} 
 
\newcommand{\mcF}{{\mathscr F}}

\newcommand{\cK}{{\mathcal K}} 
\newcommand{\cM}{{\mathcal M}} 
 
\newcommand{\cH}{{\mathcal H}} 

\newcommand{\lmin}{{\lambda_{\text{min}}}}

\newcommand{\ps}[1]{\langle #1 \rangle}

\newcommand{\ld}{\log\det} 
\newcommand{\dr}{\mathrm{dist}} 
\newcommand{\Zg}[1]{\Z_{\gamma,#1}} 
\newcommand{\Cb}{C_{\text{b}}} 

\newcommand{\bs}{\boldsymbol}
\newcommand{\dd}{{\mathsf d}}
\newcommand{\f}{{\mathsf f}}
\newcommand{\p}{{\mathsf p}}
\newcommand{\B}{{\mathsf B}}
\newcommand{\F}{{\mathsf F}}
\newcommand{\G}{{\mathsf G}} 
\newcommand{\W}{{\mathsf W}}  
\renewcommand{\H}{{\mathsf H}}  
\newcommand{\Z}{{\mathsf Z}}  
\newcommand{\X}{{\mathsf X}}  
\newcommand{\U}{{\mathsf U}}  
\newcommand{\V}{{\mathsf V}}  
\newcommand{\Ith}{\hat{\mathcal I}_n^{\mathrm{\,Th1}}}

\begin{document}

\title{Mutual Information of Wireless Channels\\
and Block-Jacobi Ergodic Operators }

\author{
Walid Hachem\footnote{
CNRS / LIGM (UMR 8049), Universit\'e Paris-Est Marne-la-Vall\'ee, France. 
Email: {walid.hachem@u-pem.fr},} \and  
Adrien Hardy\footnote{Laboratoire Paul Painlev\'e, Universit\'e de Lille, France. 
  Email: {{adrien.hardy@univ-lille.fr}},} \and 
Shlomo Shamai (Shitz)\footnote{
 Technion - Israel Institute of Technology. 
  Email: {{sshlomo@ee.technion.ac.il}}} 
}


\date{8 May 2019} 

\maketitle

\begin{abstract} Shannon's mutual information of a random multiple antenna and
multipath time varying channel is studied in the general case where the process
constructed from the channel coefficients is an ergodic and stationary process
which is assumed to be available at the receiver.  From this viewpoint, the
channel can also be represented by an ergodic self-adjoint block-Jacobi
operator, which is close in many aspects to a block version of a random
Schr\"odinger operator.  The mutual information is then related to the
so-called density of states of this operator. In this paper, it is shown that
under the weakest assumptions on the channel, the mutual information can be
expressed in terms of a matrix-valued stochastic process coupled with the
channel process. This allows numerical approximations of the mutual information
in this general setting.  Moreover, assuming further that the channel
coefficient process is a Markov process, a representation for the mutual
information offset in the large Signal to Noise Ratio regime is obtained in
terms of another related Markov process.  This generalizes previous results
from Levy~\emph{et.al.} \cite{lev-som-sha-zei-09,lev-zei-sha-it10}.  It is also
illustrated how the mutual information expressions that are closely related to
those predicted by the random matrix theory can be recovered in the large
dimensional regime.  
\end{abstract}

{\bf Keywords : } Ergodic Jacobi operators, Ergodic wireless channels, 
Large random matrix theory, Markovian channels, Shannon's mutual information. 

\section{Introduction}

In order to introduce the problem that we shall tackle in this paper, we
consider the example of a wireless communication model on a time and frequency
selective channel that is described by the equation 
\begin{equation}
\label{siso} 
y_n = \sum_{\ell=0}^L c_{n,\ell} s_{n-\ell} + v_n , 
\end{equation} 
where $L$ is the channel degree, where the complex numbers $s_n$, $y_n$ and
$v_n$ represent respectively the transmitted signal, the received signal, and
the additive noise at the moment $n$, and where the vector $C_n = [ c_{n,0},
\ldots, c_{n,L} ]^\T \in \CC^{L+1}$ contains the channel's coefficients 
at the moment $n$. In a mobile environment, the sequence $(C_n)$ is often
modeled as a random ergodic process such as $\EE\| C_0 \|^2 < \infty$
(here we take $\| \cdot \|$ as the Euclidean norm). Assuming that this process 
is available at the receiver site, our purpose is to study Shannon's mutual
information of this channel under the generic ergodicity assumption. By
stacking $n-m+1$ elements of the received signal, where $m,n \in \ZZ$ and 
$m \leq n$, we get the vector model $\begin{bmatrix} y_m, \ldots, y_n
\end{bmatrix}^\T = \B_{m,n} \begin{bmatrix} s_{m-L}, \ldots, s_n
\end{bmatrix}^\T + \begin{bmatrix} v_m, \ldots, v_n \end{bmatrix}^\T$ with 
\[ 
\B_{m,n} = \begin{bmatrix} 
c_{m,L} & \cdots & c_{m,0} \\ 
        & \ddots &        & \ddots \\ 
        &        & c_{n,L} & \cdots & c_{n,0} 
\end{bmatrix} . 
\]
Let $\rho > 0$ be a parameter that represents the Signal to Noise Ratio (SNR). 
Considering the matrix/vector model above, and putting some standard 
assumptions on the statistics of the processes $(s_n)$ and $(v_n)$ (see below),
this mutual information is written as 
\begin{equation} 
\label{limps} 
{\mathcal I}_\rho = 
\aslim_{n-m\to\infty} 
 \frac{\log\det(\rho \B_{m,n} \B_{m,n}^* + I_{n-m+1})}{n-m+1}
= 
\lim_{n-m\to\infty} 
 \frac{\EE \log\det(\rho \B_{m,n} \B_{m,n}^* + I_{n-m+1})}{n-m+1}, 
\end{equation} 
where $\B_{m,n}^*$ is the matrix adjoint of $\B_{m,n}$, and 
where the existence and the equality of both the limits above 
(``$\aslim$'' stands for the almost sure limit) are essentially due to the 
ergodicity of $(C_n)$. 

The natural mathematical framework for studying this limit is provided 
by the ergodic operator theory in the Hilbert space $\ell^2(\ZZ)$, for whom
a very rich literature has been devoted in the field of statistical physics
\cite{pas-fig-92}. In our situation, $\B_{m,n}$ is a finite rank truncation of
the operator $\B$ represented by the doubly infinite matrix 
\[ 
 \B = \begin{bmatrix} 
  \ddots &         & \ddots \\  
         & c_{n,L} & \cdots & c_{n,0} \\ 
         &         & \ddots &         & \ddots   
 \end{bmatrix} . 
\]
Thanks to the ergodicity of $(C_n)$, it is known that the spectral measure
(or eigenvalue distribution) of the matrix $\B_{m,n} \B_{m,n}^*$ converges 
narrowly in the almost sure sense to a deterministic probability measure 
called the density of states of the self-adjoint operator $\B \B^*$, where
$\B^*$ is the adjoint of $\B$. This convergence leads to the convergences 
in~\eqref{limps}. 

In statistical physics, the study of the density of states has focused most
frequently on the Jacobi (or tridiagonal) ergodic operators which are
associated to the so-called discrete Schr\"{o}dinger equation in a random
environment. In this framework, the Herbert-Jones-Thouless formula
\cite{car-lac-90, pas-fig-92} provides a means of characterizing the density of
states of an ergodic Jacobi operator,  in connection with the so-called
Lyapounov exponent associated with a certain sequence of matrices. 

In the context of the wireless communications that is of interest here, it
turns out that the use of the Thouless formula is possible when one
considers $\B\B^*$ as a block-Jacobi operator. This idea was developed by
Levy~\emph{et al.} in~\cite{lev-zei-sha-it10}.  The expression of the mutual
information that was obtained in~\cite{lev-zei-sha-it10} was also used to 
perform a large SNR asymptotic analysis so as to obtain bounds on the mutual
information in this regime.   

In this paper, we take another route to calculate the mutual information.  The
expression we obtain for $\mathcal I_\rho$ in Theorem~\ref{th-main} below
involves an ergodic process which is coupled with the channel process, and
appears to be more tractable than the expression based on the top Lyapounov
exponent provided in \cite{lev-zei-sha-it10}.  We moreover exploit the obtained
expression for $\mathcal I_\rho$ to study two asymptotic regimes: we first
consider the large SNR regime in a Markovian setting, and obtain an exact
representation for the constant term in the expansion of ${\mathcal I}_\rho$
for large $\rho$. We also consider a regime where the dimensions of the blocks
of our block-Jacobi operator converge to infinity; the expression of the mutual
information that we recover is then closely related to what is obtained from
random matrix theory \cite{kho-pastur93,hmp-jmp15}.  In the context of the
example described by Equation~\eqref{siso}, this asymptotic regime amounts to
$L$ converging to infinity. Beyond this example, the large dimensional analysis
can also be used to analyze the behavior of the mutual information of time and
frequency selective channels in the framework of the massive Multiple Input
Multiple Output (MIMO) systems (\cite{mar-etal-(livre)16}), which are destined
to play a dominant role in the future wireless cellular techniques/standards. 

\paragraph*{Organisation of the paper.} In Section~\ref{sec-pb}, after stating
precisely our communication model and our standing assumption, we provide our
main result (Theorem~\ref{th-main}). We then consider the large SNR regime in a
Markovian setting (Theorem~\ref{th-Markov}) along with some cases where the
assumptions for this theorem to hold true are satisfied.  In
Section~\ref{sec:numerics} we illustrate Theorems~\ref{th-main}
and~\ref{th-Markov} with numerical experiments. There we also state our result
on the large dimensional regime, which is related with one of the channel
models considered in this section.  The next sections are devoted to the
proofs. 

\section{Problem description and statement of the results} 
\label{sec-pb} 

\subsection{The model}
The model herein is well-suited for the block-Jacobi formalism that we use in
the remainder.  Given two positive integers $N$ and $K$, we consider the
wireless transmission model 
\begin{equation}
\label{model} 
Y_n = \F_n S_{n-1} + \G_n S_n + V_n
\end{equation} 
with $n\in\ZZ$ and where:
\begin{itemize} 
\item[-] $(Y_n)_{n\in\ZZ}$ represents the $\CC^N$-valued sequence of received signals.
\item[-] $(S_n)_{n\in\ZZ}$ is the $\CC^K$-valued sequence of transmitted information symbols.
\item[-]$(\F_n, \G_n)_{n\in\ZZ}$ with $\F_n, \G_n \in \CC^{N\times K}$ is a
matrix 
representation of a random wireless channel.
\item[-]  $(V_n)_{n\in\ZZ}$ is the additive noise. 
\end{itemize}
Let us first give a few examples which fit with this transmission model. 

\paragraph{The multipath single antenna fading channel.} The channel 
described by Equation~\eqref{siso} is a particular case of this model. 
When $L > 0$, we put 
\begin{equation} 
\label{YSV} 
Y_n := \begin{bmatrix} y_{nL} \\ \vdots \\ y_{nL+L-1} \end{bmatrix}, \quad
S_n := \begin{bmatrix} s_{nL} \\ \vdots \\ s_{nL+L-1} \end{bmatrix}, \quad
V_n := \begin{bmatrix} v_{nL} \\ \vdots \\ v_{nL+L-1} \end{bmatrix}, \quad N:=K:=L,
\end{equation}  
and $\F_n,\G_n\in\CC^{L\times L}$ are  the upper triangular and lower triangular
matrices  defined as 
\begin{equation} 
\label{FGmultipath} 
\left[\begin{array}{c|c} \F_n & \G_n \end{array}\right] := 
\left[\begin{array}{ccc|ccc} 
c_{nL,L} & \cdots & c_{nL,1}     & c_{nL,0}       &        &              \\
         & \ddots &  \vdots      & \vdots         & \ddots &              \\
         &        & c_{nL+L-1,L} & c_{nL+L-1,L-1} & \cdots & c_{nL+L-1,0}
\end{array}\right] . 
\end{equation} 
When $L=0$,  we set instead $N:=K:=1$, $Y_n := y_n$, $S_n := s_n$, $V_n := v_n$, $\F_n :=
0$, and $\G_n := c_{n,0}$. 

In the multiple antenna variant of this model, the channel coefficients
$c_{n,\ell}$ are $R\times T$ matrices, where $R$, resp.~$T$, is the number of
antennas at the receiver, resp.~transmitter. In this case, the $N\times
K$ matrices $\F_n$ and $\G_n$ given by Eq.~\eqref{FGmultipath} when $L > 0$ 
are block triangular matrices with $N := RL$ and $K:=TL$.

\paragraph{The Wyner multi-cell model.} Another instance of the transmission
model introduced above is a generalization of the so-called Wyner multi-cell
model considered in \cite{han-whi-93,wyn-it94}, where the index $n$ now
represents the space instead of representing the time. Assume that the Base
Stations (BS) of a wireless cellular network are arranged on a line, and that
each BS receives in a given frequency slot the signals of the $L+1$ users which
are not too far from this BS. Alternatively, each user is also seen by $L+1$
BS. In this setting, the signal $y_n$ received by the BS $n$ is described by
Eq.~\eqref{siso} (where the time parameter is now omitted), where $s_n$ is the
signal emitted by User~$n$, and where $c_{n,\ell}$ is the uplink channel 
carrying the signal of User~$n-\ell$ to BS $n$. \\

Other domains than the time or the space domain, such as the frequency domain,
can also be covered, see \emph{e.g.} \cite{tul-cai-sha-ver-10}, which deals
with a time and frequency selective model.  Moreover, this could even address
different connected domains as the Doppler-Delay (connected via the so-called
Zak transform), as in~\cite{bol-hla-sp97,bol-duh-hle-elssp03}, which lead to
modulation schemes that are considered as interesting candidates for the fifth
generation (5G) wireless systems, as reflected in the references
\cite{had-etal-wcnc17,cai-etal-comsurv18}. 

\subsection{General assumptions}
\label{sec:Assgeneral} 
The purpose of this work is to study Shannon's mutual information between
$(S_n)$ and $(Y_n)$ when the channel is known at the receiver. To this end, we
consider the usual setting where:
\begin{itemize}
\item[-] The information sequence $(S_n)_{n\in\ZZ}$ is random i.i.d.~with law 
$\mathcal{CN}(0,I_K)$.
\item[-]  The noise $(V_n)_{n\in\ZZ}$ is i.i.d. with law 
  $\mathcal{CN}(0,\rho^{-1} I_N)$ for some $\rho>0$ that scales with the SNR. 
\item[-] The random sequences $(S_n)_{n\in\ZZ}$, $(\F_n, \G_n)_{n\in\ZZ}$, and
$(V_n)_{n\in\ZZ}$ are independent.
\end{itemize}
Here and in the following, i.i.d.~means ``independent and identically
distributed'', and $\mathcal{CN}(0,\Sigma)$ stands for the law of a centered
complex Gaussian circularly symmetric vector with covariance matrix $\Sigma$. 
We also make the
following assumptions on the  process $(\F_n, \G_n)_{n\in\ZZ}$ representing the
channel:
\begin{assumption}
\label{Ass1}
 The process $(\F_n, \G_n)_{n\in\ZZ}$
 is a \emph{stationary} and \emph{ergodic} process. Moreover,  
\begin{equation}
\label{momFG} 
\EE \| \F_0 \|^2 < \infty\quad\text{ and }\quad \EE \| \G_0 \|^2 < \infty . 
\end{equation}
\end{assumption}
Note that the moment assumption~\eqref{momFG} does not depend on the specific
choice of the norm on the space of $N\times K$ complex matrices. In the
remainder, we choose $\|\cdot\|$ to be the spectral norm. 

Let us make precise the assumptions of stationarity and ergodicity. In the following we set for convenience 
\begin{equation} 
E:= \CC^{N\times K}\times  \CC^{N\times K}
\end{equation} 
and consider the measure space $\Omega := E^\ZZ$ equipped 
with its Borel $\sigma$--field $\mcF := \mcB(E)^{\otimes \ZZ}$. An element of $\Omega$ reads 
$\omega = \left( \ldots, (F_{-1}, G_{-1}), (F_{0}, G_{0}), (F_{1}, G_{1}), 
\ldots \right)$
where $(F_n, G_n)$ is the $n^{\text{th}}$ coordinate of $\omega$, with
$(F_n,G_n)\in E$. The shift $\T:\Omega\to\Omega$ acts as 
$
\T\omega := 
 \left( \ldots, (F_{0}, G_{0}), (F_{1}, G_{1}), (F_{2}, G_{2}), 
\ldots \right).
$
The assumption that $(\F_n, \G_n)_{n\in\ZZ}$ is an ergodic stationary process,
seen as a measurable map from $(\Omega,\mcF)$ to itself, means that the shift
$\T$ is a measure preserving and ergodic transformation with respect to the
probability distribution of the process $(\F_n, \G_n)_{n\in\ZZ}$. 

A fairly general stationary and ergodic model is provided by the following
example. 
\begin{example}
\label{ex:AR} 
In the single antenna and single path ($L=0$) fading channel case, the
autoregressive (AR) statistical model has been considered as a realistic model
for representing the Doppler effect induced by the mobility of the
communicating devices. This model reads 
\begin{equation}
c_{n,0} = \sum_{\ell=1}^M a_\ell c_{n-\ell, 0} + u_n, 
\end{equation} 
where $M > 0$ is the order of the AR channel process, $(u_n)_{n\in\ZZ}$ is
an i.i.d.~driving process, and $(a_1,\ldots, a_M)$ are the constant AR filter 
coefficients, which can be tuned to meet a required Doppler spectral density 
(see, \emph{e.g.},~\cite{bad-bea-05}). 

In the multipath case, this model can be generalized to account for the
presence of a power delay profile and the presence of correlations between the
channel taps in addition to the Doppler effect. In this case, the channel
coefficients vector $C_n = [ c_{n,0},\ldots, c_{n,L}]^\T$ is written as 
\begin{equation}
\label{channel-ar} 
C_n = \sum_{\ell=1}^M A_\ell C_{n-\ell} + U_n , 
\end{equation} 
where $\{A_1,\ldots, A_M\}$ is a collection of deterministic 
$(L+1)\times (L+1)$ matrices, and where $(U_n)_{n\in\ZZ}$ is a
$\CC^{L+1}$--valued i.i.d.~driving process. If the polynomial 
$\det(I - \sum_{\ell=1}^M z^\ell A_\ell)$ does not vanish in the closed unit
disc, it is well known that there exists a stationary and ergodic process whose
law is characterized by \eqref{channel-ar}, see
\emph{e.g.}~\cite{kai-(livre)80,mey-twe-livre09}, leading to a stationary and
ergodic process $(\F_n,\G_n)_{n\in\ZZ}$ by recalling the construction of
$\begin{bmatrix} \F_n \, | \, \G_n \end{bmatrix}$ given by 
Equation~\eqref{FGmultipath}. 
\end{example}

\subsection{Mutual information and statement of the main result} 
In order to define the mutual information of the channel described
by~\eqref{model}, define for any $m,n\in\ZZ$, $m\leq n$, the random matrix of
size $(n-m+1)N \times (n-m+2)K$,
\begin{equation}
\label{Hmn} 
\H_{m,n} := \begin{bmatrix} 
\F_{m} & \G_{m} \\ 
      & \F_{m+1} & \G_{m+1}                  \\
      &         & \ddots & \ddots \\ 
      &         &        & \F_{n} & \G_{n} 
\end{bmatrix}  . 
\end{equation} 
For any fixed $\rho>0$, let ${\mathcal I}_\rho$ be given by 
\begin{align} 
\mathcal I_\rho &:=
\aslim_{n-m\to\infty}
\frac{1}{(n-m+1)N} \log\det\left(I+\rho\, \H_{m,n} \H_{m,n}^* \right) 
 \nonumber \\
&= \lim_{n-m\to\infty}
\frac{1}{(n-m+1)N} \EE \log\det\left(I+\rho\, \H_{m,n} \H_{m,n}^* \right) . 
\label{defI}
\end{align} 
As we shall briefly explain below, these two limits exist, are finite and
equal, and do not depend on the way $n-m\to\infty$ due to the
Assumption~\ref{Ass1}. As is well known, $\mathcal I_\rho$ is known to
represent the required mutual information per component of our wireless
channel, provided the input $S_n$ is as in Section~\ref{sec:Assgeneral},
see~\cite{gray-entropy11}.  The purpose of this paper is to study this
quantity. 

\begin{remark}
In the Wyner multicell model introduced above, where the BS collaborate
while the users do not, ${\mathcal I}_\rho$ represents the sum mutual 
information per component. 
\end{remark} 

Denoting by  $\cH_K^{++}$, resp. $\cH_K^{+}$, the cone of the Hermitian positive definite, resp. semidefinite,  $K\times
K$ matrices, we show that one can construct  a stationary  $\cH_K^{++}$-valued
process $(\W_n)_{n\in\ZZ}$ defined recursively and coupled with
$(\F_n,\G_n)_{n\in\ZZ}$ which allows a rather simple formula for the mutual
information per component $\mathcal I_\rho$.  

\begin{theorem}[Mutual information of an ergodic channel] 
If Assumption~\ref{Ass1} holds true, then:
\label{th-main} \begin{itemize} 
\item[{\rm(a)}]There exists a unique stationary  $\cH_K^{++}$-valued process $(\W_n)_{n\in\ZZ}$ satisfying
\begin{equation}
\label{recurs} 
\W_n = \left( I + \rho \, \G_n^*  \left( I + \rho \,\F_n \W_{n-1} \F_n^* \right)^{-1} \G_n \right)^{-1}.
\end{equation}
In particular, the process $(\W_n)$ is  ergodic.
\item[{\rm(b)}] We have the representation for the mutual information per
component:
\begin{equation} 
\label{I-WW} 
\mathcal I_\rho  = \frac1N\Big( \EE\log\det\left( I + \rho \,\F_0 \W_{-1} \F_0^* \right)-\EE \log\det \W_0  \Big).
\end{equation} 
\item[{\rm (c)}] Given \emph{any} matrix $\X_{-1} \in \cH_K^{+}$, if one defines a process 
$(\X_n)_{n\in\NN}$ by setting
\begin{equation} 
\X_n := \left( I + \rho \, \G_n^*  \left( I + \rho \,\F_n  \X_{n-1} \F_n^* 
  \right)^{-1} \G_n \right)^{-1}
\end{equation} 
for all $n\geq 0$, then we have
\begin{equation} 
\mathcal I_\rho = \lim_{n\to\infty} \frac{1}{nN} \sum_{\ell=0}^{n-1} 
 \log\det \left( I + \rho \,\F_\ell \X_{\ell-1} \F_\ell^* \right) - \log\det \X_\ell \quad \text{a.s.}
\end{equation} 
 \end{itemize}
\end{theorem} 

The proof of Theorem~\ref{th-main} is provided in Section~\ref{sec:prfmain}.

\begin{remark} As we will illustrate in Section~\ref{sec:numerics}, Theorem~\ref{th-main}(c) yields an estimator for $\mathcal I_\rho$ that is less costly numerically than the naive one, due to the dimension of the involved matrices. 
\end{remark}
\begin{remark}
\label{weak-mom} 
The proof of Theorem~\ref{th-main} reveals that the moment 
assumption~\eqref{momFG} can be weakened~to 
\begin{equation} 
\EE\log( 1 + \| \F_0 \|^2) < \infty \qquad\text{ and }\qquad \EE\log( 1 + \| \G_0 \|^2 ) < \infty . 
\end{equation} 
The second moment assumption \eqref{momFG} is here to
ensure that the received signal power is finite. 
\end{remark} 

\begin{remark} 
An expression for ${\mathcal I}_\rho$ similar to the one given by
Theorem~\ref{th-main} is obtained by Levy~\emph{et.al.}
in~\cite{lev-som-sha-zei-09} in the particular case where $N=1$ and
where the process $(\F_n, \G_n)$ is i.i.d. 
\end{remark} 

\subsection{Connection to  block-Jacobi operators and previous results} 

Recall Eq.~\eqref{Hmn}. Due to Assumption~\ref{Ass1}, it is
well known, see~\cite{pas-fig-92}, that there exists a deterministic
probability measure $\mu$ that can defined by the fact that for each bounded
and continuous function $f$ on $[0,\infty)$,
\begin{equation}
\label{cvg-ids} 
\frac{1}{(n-m+1)N} \tr f(\H_{m,n} \H_{m,n}^*) 
\xrightarrow[n-m\to\infty]{} \int f(\lambda) \, \mu(d\lambda) 
\quad \text{a.s.} 
\end{equation} 
(here, $f$ is of course extended by functional calculus to the semi-definite
positive matrices). 
The measure $\mu$ is intimately connected with the so-called
\emph{ergodic self-adjoint block-Jacobi (or block-tridiagonal) operator}
$\H \H^*$, where $\H$ is the random linear operator acting on the Hilbert space 
$\ell^2(\ZZ)$, and defined by its doubly-infinite matrix representation 
in the canonical basis $(e_k)_{k\in\ZZ}$ of this space as 
\begin{equation}
\label{def-H} 
\H = \begin{bmatrix} 
\ \ddots & \ddots  \\ 
       & \F_{-1} & \G_{-1} \\ 
       &         & \F_0     & \G_0                  \\
       &         &         & \F_{1} & \G_{1}    \\ 
       &         &         &         &\ddots & \ddots \ 
\end{bmatrix}.
\end{equation}
The random positive self-adjoint operator $\H \H^*$ is an ergodic operator in
the sense of~\cite[Page 33]{pas-fig-92} (see also \cite{hmp-jmp15}), and the
measure $\mu$ is called its \emph{density of states}. Recalling~\eqref{defI},
it holds that 
\begin{equation} 
\mathcal I_\rho=\int\log(1+\rho \lambda)\,\mu(d\lambda), 
\end{equation} 
where this limit is finite, due to the moment assumption~\eqref{momFG} and a
standard uniform integrability argument.  

As said in the introduction, the Herbert-Jones-Thouless
formula~\cite{car-lac-90, pas-fig-92} provides a means of characterizing the
density of states of an ergodic Jacobi operator.  In \cite{lev-zei-sha-it10},
Levy~\emph{et al.}~develop a version of this formula that is well suited to the
block-Jacobi setting of $\H \H^*$. 

In this paper, we rather identify $\mathcal I_\rho$ by considering the
resolvents of certain random operators built from the process
$(\F_n,\G_n)_{n\in\ZZ}$ instead of using the Herbert-Jones-Thouless formula.
The expression we obtain for $\mathcal I_\rho$ involves the ergodic process
$(\W_n)$ which is coupled with the process $(\F_n,\G_n)_{n\in\ZZ}$ by
Eq.~\eqref{recurs}. This approach is developed in 
Section~\ref{sec:prfmain}.

\subsection{The Markovian case and large SNR regime}

First,  assuming extra assumptions on the process $(\F_n,\G_n)$, we obtain a
description for the constant term (or mutual information offset) in the large
SNR regime. Indeed, it often happens that there exists a real number
$\kappa_\infty$ such that the mutual information per component admits the 
expansion  as $\rho\to\infty$,
\begin{equation} 
\mathcal{I}_\rho = 
 \min(K / N,1) \log\rho + \kappa_\infty + o(1) , 
\end{equation} 
see e.g. \cite{loz-tul-ver-it05}.  Our next task is to prove this expansion
indeed holds true and to derive an expression for the offset 
$\kappa_\infty$ when the process $(\F_n,\G_n)_{n\in\ZZ}$ is further assumed to
be a Markov process  satisfying some regularity and moment assumptions. Namely,
consider for any $n\in\ZZ$ the $\sigma$-field $\mathscr F_n:=\sigma( (\F_k,
\G_k) \; : k\leq n)$ and assume there exists a transition kernel $P : E \times
\mcB(E) \to [0,1]$ such that, for any Borel function $f:E \to [0,\infty)$,  
\begin{equation} 
\EE\big[f(\F_{n+1},\G_{n+1}) | \mathscr F_n\big]=Pf((\F_n, \G_n)):= \int f(F, G) \, P\big((\F_n, \G_n), dF \times dG\big).
\end{equation} 
Besides $Pf((\F, \G))$, we use the common notations from the Markov chains
literature and also write $P((\F, \G),A):=P\bs 1_A((\F, \G))$ for any Borel set
$A\in\mcB(E)$; the iterated kernel $P^n$ stands for the Markov kernel defined
inductively by $P^nf:=P(P^{n-1}f)$ with the convention that $P^0f:=f$; given
any probability measure $\eta$ on $E$, we let $\eta P$ be the probability
measure on $E$ defined as 
\begin{equation} 
\eta P(A):=\int P((F,G),A)\,\eta(dF\times dG),\qquad A\in\mcB(E).
\end{equation}

The following assumption is formulated in the context where $N > K$. We denote
as $\cM(E)$ the space of Borel probability measures on the space $E$.  Given a
matrix $A$, the notations $\Pi_A$ and $\Pi_A^\perp$ refer respectively to the
orthogonal projector on the column space $\colspan(A)$ of $A$, and to the
orthogonal projector on $\colspan(A)^\perp$.  

\begin{assumption}
\label{Ass2}
 The process $(\F_n, \G_n)_{n\in\ZZ}$
 is a Markov process with transition kernel $P$ associated with a unique invariant probability measure $\theta\in\cM(E)$, namely satisfying $\theta P=\theta$. Moreover,  
 \begin{itemize}
 \item[(a)]$P$ is Feller, namely, if $f: E\to\RR$ is continuous and bounded, 
  then so is $Pf$.
\item[(b)]  $\EE \| \F_0 \|^2 + \EE \| \G_0 \|^2  < \infty$.

\item[(c)] $\EE|\log\det(\F_0^* \F_0)|<\infty$.

\item[(d)]
For  every non-zero $v \in \CC^K$,  we have for $\theta$-a.e. $(F,G)\in E$ that
\begin{equation} 
\label{regCond}
\det( G^* F) \neq 0 
\qquad \text{and} \qquad 
 \Pi_{G}^\perp F v \neq 0 \, . 
\end{equation} 
\end{itemize}
\end{assumption}

\begin{remark}
Since a Markov chain $(\F_n,\G_n)_{n\in\ZZ}$ associated with a unique invariant probability measure is automatically ergodic, we see that Assumption~\ref{Ass2} is stronger than Assumption~\ref{Ass1} and thus Theorem~\ref{th-main} applies in this setting.
\end{remark}

\begin{remark} If one assumes $(\F_n,\G_n)_{n\in\ZZ}$ is a sequence of i.i.d random variables with law $\theta$ having a density on $E$, then it satisfies Assumption~\ref{Ass2}  (and hence Assumption~\ref{Ass1}) provided that the moment conditions Assumption~\ref{Ass2}(b)-(c) are satisfied. We also provide more sophisticated examples were Assumption~\ref{Ass2} holds in Section~\ref{Examples}.
\end{remark}

\begin{remark}
Since  $\theta=\theta P$, Assumption~\ref{Ass2}(d) equivalently says that, for $\theta$-a.e. $(\F,\G)$, \eqref{regCond} holds true for $P((\F,\G),\cdot)$-a.e. $(F,G)\in E$. We will use this observation at several instances in the following. 
\end{remark}

\begin{theorem}[The Markov case] 
\label{th-Markov}
Let $N > K$. Then, under Assumption~\ref{Ass2}, the following hold true: 
\begin{itemize}
\item[(a)] There exists a unique stationary process $(\Z_n)_{n\in\ZZ}$ on $\cH_K^{++}$ satisfying 
\begin{equation} 
\label{eq:Zn}
\Z_n=\G_n^*(I+\F_n\Z_{n-1}^{-1}\F_n^*)^{-1}\G_n .
\end{equation} 
\item[(b)] We have, as $\rho\to\infty$,
\begin{equation} 
\mathcal I_\rho = \frac KN \log\rho + \kappa_\infty + o(1),  
 \end{equation} 
 where $\log\det(\Z_0+\F_1^*\F_1)$ is integrable, and 
\begin{equation} 
  \kappa_\infty:=\frac1N\EE\log\det(\Z_0+\F_1^*\F_1).
\end{equation} 
\item[(c)] Given \emph{any} $\X_{-1} \in \cH_K^{++}$, if we consider the process 
$(\X_n)_{n\in\NN}$ defined recursively by 
\begin{equation} 
\X_{n} = \G_n^*(I+\F_n\X_{n-1}^{-1}\F_n^*)^{-1}\G_n\,,
\end{equation} 
then we have, in probability, 
\begin{equation}
\label{Xkappa} 
 \kappa_\infty=\lim_{n\to\infty}\frac 1{nN} \sum_{\ell=0}^{n-1} \log\det(\X_{\ell} + \F_{\ell+1}^* \F_{\ell+1}).
\end{equation}

\end{itemize}
\end{theorem}
The proof of Theorem~\ref{th-Markov} is provided in  Section~\ref{sec:prfasymp}.

\begin{remark}[The  case $N\leq K$] In the statement of Theorem~\ref{th-Markov}, it is assumed that $N > K$. Let
us say a few words about the case where $N < K$. In this case, 
assuming that $(\F_n, \G_{n-1})$ is a Markov chain, there is an analogue 
$(\widetilde \Z_n)$ of the process $(\Z_n)$ satisfying the recursion 
\begin{equation} 
\widetilde \Z_n=\F_n
    (I_K+\G_{n-1}^* \widetilde \Z_{n-1}^{-1}\G_{n-1})^{-1}\F_n^*, 
\end{equation} 
and adapting Assumption~\ref{Ass2} to this new setting, we can show that 
${\mathcal I}_\rho = \log\rho + \tilde\kappa_\infty + o(1)$, where 
\begin{equation} 
\tilde\kappa_\infty := \frac1N \EE \log\det(\widetilde\Z_0 + \G_0 \G_0^*) .  
\end{equation}  
This result can be obtained by adapting the proof of Theorem~\ref{th-Markov} 
in a straightforward manner. \\ 
The case $K=N$ is somehow singular and requires a specific treatment that will not be undertaken 
in this paper; see also the end of Section~\ref{sec:Qgamma} for
further explanations. 
\end{remark} 

\begin{remark}
In the case where $K=1$, $N > 1$, and the process $(\F_n, \G_n)_{n\in\ZZ}$ is 
i.i.d., we recover \cite[Th.~2]{lev-som-sha-zei-09}, where this result is obtained
with the help of the theory of Harris Markov chains. 
\end{remark} 

\subsubsection*{Examples where Assumption~\ref{Ass2} is verified} 
\label{Examples} 

In Proposition~\ref{prop:com1} below, the Markov property of the process
$(\F_n, \G_n)_{n\in\ZZ}$ is obvious, while in Proposition~\ref{prop:com2}, it
can be easily checked from Equation~\eqref{FGmultipath}. Moreover, in 
both propositions, it is well known that the Markov process 
$(\F_n, \G_n)_{n\in\ZZ}$ is an ergodic process satisfying Assumptions
\ref{Ass2}-(a) and \ref{Ass2}-(b)~\cite{mey-twe-livre09}. We shall focus on
Assumptions~\ref{Ass2}-(c) and \ref{Ass2}-(d). 

\begin{proposition}[AR-model]
\label{prop:com1} 
For $N > K$, assume $(\F_n, \G_n)$ is the multidimensional ergodic AR process 
defined by the recursion 
\begin{equation} 
\begin{bmatrix} \F_{n} \\ \G_{n} \end{bmatrix} 
= 
A \begin{bmatrix} \F_{n-1} \\ \G_{n-1} \end{bmatrix} + 
\begin{bmatrix} U_{n} \\ V_{n} \end{bmatrix} , 
\end{equation}  
where $A \in \CC^{2N\times 2N}$ is a deterministic matrix whose eigenvalue 
spectrum belongs to the open unit disk, and where $(U_n, V_n)_{n\in\ZZ}$ is an 
i.i.d.~process on $E$ such that $\EE \| U_0 \|^2 + \EE \| V_0 \|^2 < \infty$. 
If the entries of the matrix $\begin{bmatrix} U_n \ V_n\end{bmatrix}$  are 
independent with their distributions being absolutely continuous with respect 
to the Lebesgue measure on $\CC$, then Assumption~\ref{Ass2}-(d) is verified. 
If, furthermore, the densities of the elements of $U_n $ and $V_n$ are bounded,
then, Assumption~\ref{Ass2}-(c) is verified.  
\end{proposition} 

Our second example is a particular multi-antenna version of the AR channel
model of Example~\ref{ex:AR}. This model is general enough to capture the
Doppler effect, the correlations within each matrix coefficient of the channel,
as well as the power profile of these taps.  

\begin{proposition}[MIMO multipath fading channel]
\label{prop:com2} 
Given three positive integers $L, R$, and $T$ such that $R > T$, let
$(C_n)_{n\in\ZZ}$ be the $\CC^{(L+1)R \times T}$-valued random process
described by the iterative model 
\begin{equation} 
\label{defCn}
C_{n} = \begin{bmatrix} H_0 \\ & \ddots \\ & & H_L \end{bmatrix}
 C_{n-1} + U_{n} , 
\end{equation} 
where the $\{ H_\ell \}_{\ell=0}^{L}$ are deterministic $R\times R$ matrices
whose spectra lie in the open unit disk, and where $(U_n)_{n\in\ZZ}$ is an 
i.i.d. matrix process such that $\EE \| U_0 \|^2 < \infty$. Let $\F_n$ and 
$\G_n$ be the $L R \times LT$ matrices defined as in \eqref{FGmultipath} with
$C_n = \begin{bmatrix} c_{n,0}^\T \ \cdots \ c_{n,L}^\T \end{bmatrix}^\T$, the
$c_{n,\ell}$'s being $R\times T$ matrices. 
If the entries of $U_n$ are independent with their distributions being 
absolutely continuous with respect to the Lebesgue measure on $\CC$, then   
Assumption~\ref{Ass2}-(d) is verified on the Markov process 
$(\F_n, \G_n)_{n\in\ZZ}$. If, furthermore, the densities of the elements of 
$U_n$ are bounded, then, Assumption~\ref{Ass2}-(c) is verified.  
\end{proposition}  

Propositions~\ref{prop:com1} and \ref{prop:com2} are proven in
Section~\ref{proofExamples}.

\section{Numerical illustrations} 
\label{sec:numerics} 

We consider here a multiple antenna version of the multipath channel desribed
in the introduction, see Equations~\eqref{YSV}--\eqref{FGmultipath}. We assume 
the channel coefficient matrices $c_{n,\ell}$ satisfy the AR model 
$c_{n,\ell} = \alpha c_{n-1,\ell} + \sqrt{1-\alpha^2} a_\ell u_{n,\ell}$. Here
the AR coefficient $\alpha$ takes the form $\alpha = \exp(- f_{\text{d}})$. 
The parameter $f_{\text{d}}$ represents the Doppler frequency, since it is
proportional to the inverse of the effective support of the autocorrelation 
function of a channel tap (channel coherence time). 
For $n \in \ZZ$ and $\ell \in \{ 0,\ldots L\}$, the $u_{n,\ell}$'s are
i.i.d. $R\times T $ random matrices with i.i.d $\mathcal{CN}(0,T^{-1})$
entries; the real vector $a = [a_0, \ldots, a_L]$ is a multipath amplitude 
profile vector such that $\| a \| = 1$; as is well known, the vector 
$[ a_0^2, \ldots, a_L^2]$  represents the so called power delay profile.  
\paragraph{Illustration of Theorem~\ref{th-main}.} We choose an exponential profile of the form 
$a_\ell \propto \exp(-0.4\ell)$. 
We start by comparing the mutual information estimates $\hat{\mathcal I}_{m,n}$ of $\mathcal I_\rho$ that naturally come with \eqref{defI}, namely by taking empirical averages of 
\begin{equation} 
 \frac{1}{(n-m+1)N} 
\, \log\det\left(I+\rho\, \H_{m,n} \H_{m,n}^* \right)
\end{equation} 
for several realizations of $\H_{m,n}$, with those coming with Theorem~\ref{th-main}(c), namely 
\begin{equation} 
\Ith:=\frac{1}{nN} \sum_{\ell=0}^{n-1} 
 \log\det \left( I + \rho \,\F_\ell \X_{\ell-1} \F_\ell^* \right) - \log\det \X_\ell
\end{equation} 
where, for any $n\in\NN$,
\begin{equation} 
\X_n := \left( I + \rho \, \G_n^*  \left( I + \rho \,\F_n  \X_{n-1} \F_n^* \right)^{-1} \G_n \right)^{-1},\qquad \X_{-1}:=I.
\end{equation}

\begin{figure}[ht]
\centering
{\includegraphics*[scale=0.7]{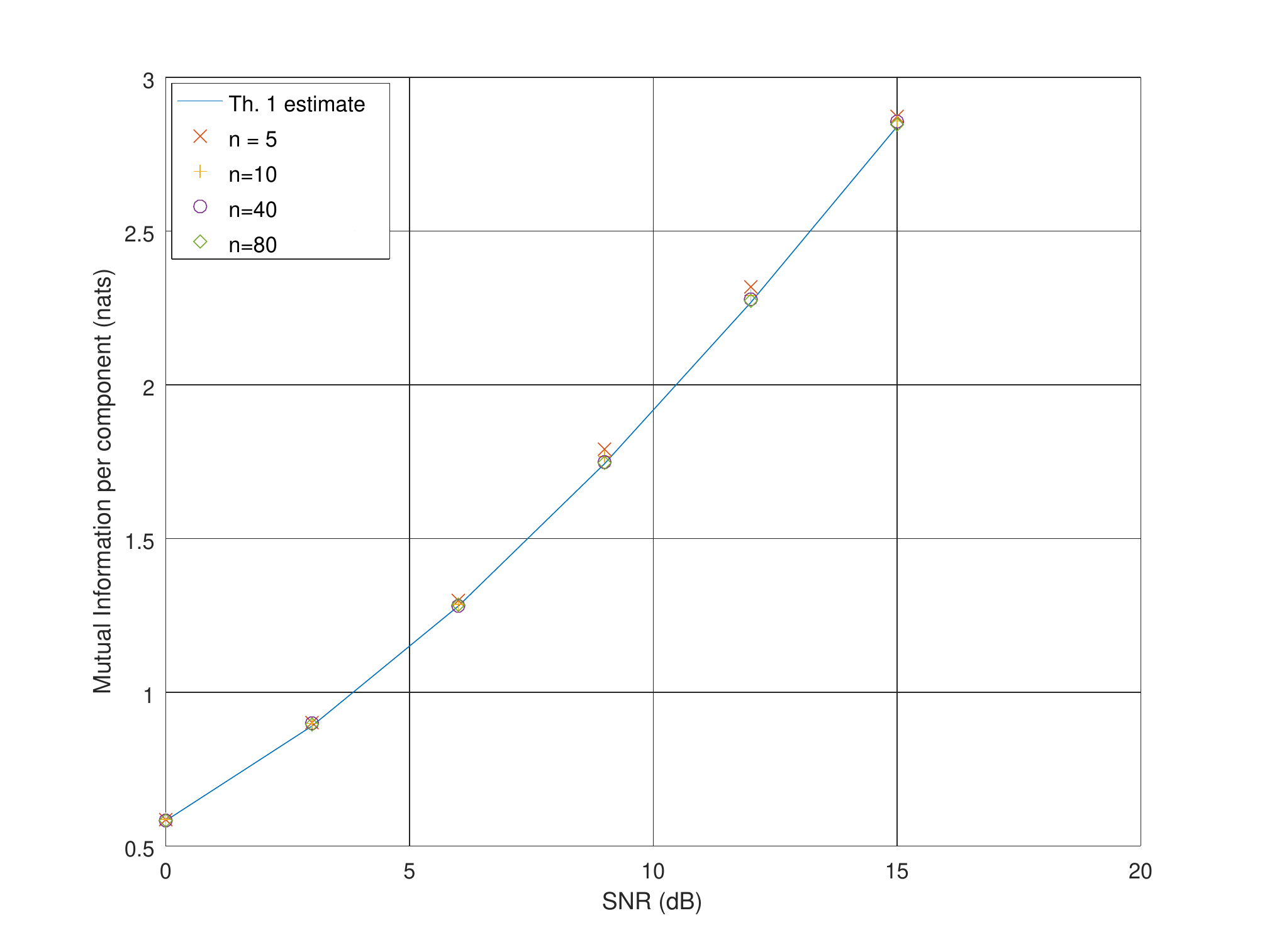}}
\caption{Plots of $\hat{\mathcal I}_{1,n}$   and $\hat{\mathcal I}_{4000}^{\mathrm{\,Th1}}$
w.r.t.~the SNR and $n$.  Setting: $R = T = 2$, $L=3$, 
$f_{\text{d}} = 0.05$.  Each empirical average $\hat{\mathcal I}_{1,n}$ comes from $150$ channel realizations.}   
\label{fig:svd-precision} 
\end{figure}
\begin{figure}[H]
\centering
{\includegraphics*[scale=0.7]{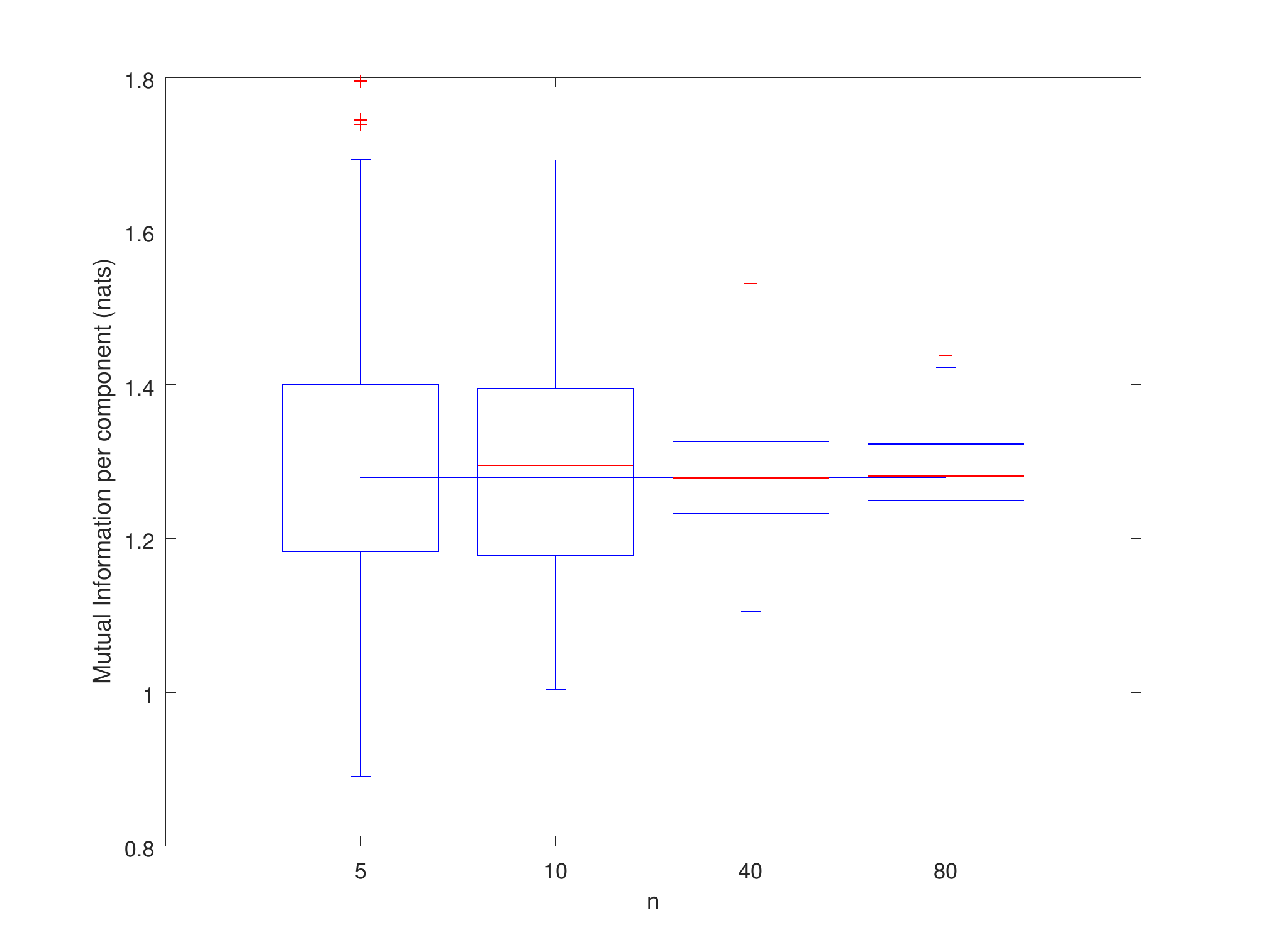}}
\caption{Boxplots  of $\hat{\mathcal I}_{1,n}$ w.r.t.~$n$.
Same setting as for Fig.~\ref{fig:svd-precision} with $\rho = 6$ dB. The 
continuous horizontal line represents $\hat{\mathcal I}_{4000}^{\mathrm{\,Th1}}$.} 
\label{fig:boxplot-precision} 
\end{figure}

Figure~\ref{fig:svd-precision} shows that the estimates of $\mathcal I_\rho$
obtained by doing empirical averages  ${\mathcal I}_{1,n}$ are not
affected by important biases. However, Figure~\ref{fig:boxplot-precision} shows
that the dispersion parameters associated with these estimates are still
important for $n$ as large as $80$. We note that in the setting of this figure,
the matrix $H_{1,n}H_{1,n}^* \in \CC^{nRL \times nRL}$ is a $480 \times 480$
matrix when $n = 80$. On the other hand, the mutual information estimates $\Ith$
provided by Theorem~\ref{th-main} require much less numerical computations
since they involve the inversions of $RL \times RL = 6 \times 6$ matrices.

\paragraph{The large random matrix regime.}

Next, we consider the asymptotic regime where both $N$ and $K$ converge to
infinity at the same pace. For a large class of processes $(\F_n, \G_n)$, it
happens that in this regime, the Density of States of the operator $\H\H^*$
(which should now be indexed by $K,N$) converges to a probability measure
encountered in the field of large random matrix theory; see~\cite{kho-pastur93}
for ``Wigner analogues'' of our model, and~\cite{hmp-jmp15} for models closer
to those of this paper.  One important feature of this probability measure is
that it depends on the probability law of the channel process only through its
first and second order statistics.

We illustrate herein this phenomenon on an instance of the MIMO frequency and
time selective channel described at the beginning of this section.  We observe
that in this applicative setting, the regime of convergence of $N,K\to\infty$
at the same rate embeds the case where $R$ and $T$ are fixed while
$L\to\infty$, the case where $L$ is fixed while $R,T \to\infty$ at the same
pace, as well as the intermediate cases.  For the simplicity of the
presentation, we assume that the numbers of antennas $R$ and $T$ are equal
(note that $N = K = RL$ in this case), and moreover, set the AR coefficient
$\alpha=0$. If we let $N\to\infty$, we get the following result: 

\begin{proposition}[large dimensional regime] 
\label{prop:mp}
Within the specific model described above, assume the vector $a$, which depends 
on $L$, satisfies $\|a\|=1$ for every $L$, and that 
\begin{equation} 
\sup_L \max_{\ell\in\{0,\ldots, L\}} \sqrt{L} | a_\ell | < \infty, 
\end{equation} 
(which is trivially satisfied if $L$ is fixed). Then, 
\begin{equation} 
\label{LargeRMT}
\lim_{N\to\infty}{\mathcal I}_\rho= 2 \log \frac{\sqrt{4\rho + 1}+1}{2} 
  - \frac{2\rho + 1 - \sqrt{4\rho + 1}}{2\rho} \, . 
\end{equation} 
\end{proposition} 
To prove this proposition, we shall show that $\mathcal I_\rho$ converges as
$N\to\infty$ to $\int\log(1+\rho\lambda) \, \mu_{\text{MP}}(d\lambda)$, where
$\mu_{\text{MP}}(d\lambda) = (2\pi)^{-1} \sqrt{4/\lambda - 1}
\1_{[0,4]}(\lambda) \, d\lambda$. This is the element of the family of the
celebrated Marchenko-Pastur distributions which is the limiting spectral
measure of $XX^*$ when $X$ is a square random matrix with iid elements. We
provide a proof in Section~\ref{sec:RMT} which is based on
Theorem~\ref{th-main}.  More sophisticated channel models can be considered,
including non centered models or models with correlations along the time index
$n$, and for which one can prove similar asymptotics, see \cite{hmp-jmp15}. 
Note also that in the context of the large random matrix theory, a 
similar model where $L$ is fixed and $R,T\to\infty$ at the same rate has been 
considered in~\cite{mul-it02}. 

We illustrate this result on an example, represented in
Figure~\ref{fig:marchenko}. As an instance of the statistical channel model
used in the statement of Proposition~\ref{prop:mp}, we assume a generalized
Wyner model as described in the introduction of this paper. We fix $R$ and $T$
to equal values, and we consider the regime where the network of Base Stations
becomes denser and denser, making $L$ converge to infinity. By densifying the
network, the number of users occupying a frequency slot will grow linearly with
the number of BS. The number of interferers will grow as well.  Yet, provided
the BS are connected through a high rate backbone to a central processing unit
which is able to perform a joint processing, the overall network capacity will
grow linearly with $L$. To be more specific, we assume that the channel power
gain when the mobile is at the distance $d$ to the BS is 
\begin{equation} 
\frac{1}{10 + (10 d/D)^3} \1_{[-D/2,D/2]}(d) , 
\end{equation} 
where $D > 0$ is a parameter that has the dimension of a distance. If the BS
are regularly spaced, and if there are $L$ Base Stations per $D$ units of
distance, then one channel model approaching this power decay behavior is the
setting where the  $a_\ell$'s are given by 
\begin{equation} 
a_\ell^2 \propto \frac{1}{10 + \left|10 (\ell - L/2)/L\right|^3} ,
\quad \ell \in \{0,\ldots,L\} .
\end{equation} 
The quantity $R\times \lim_{L\to\infty}\mathcal I_\rho$, where the limit is 
given by Proposition~\ref{prop:mp}, thus represents the ergodic mutual 
information per user. 
Figure~\ref{fig:marchenko} shows that the predictions of 
Proposition~\ref{prop:mp} fit with the values provided by Theorem~\ref{th-main}
for $L$ as small as one. 
\begin{figure}[ht]
\centering
{\includegraphics*[width=\columnwidth]{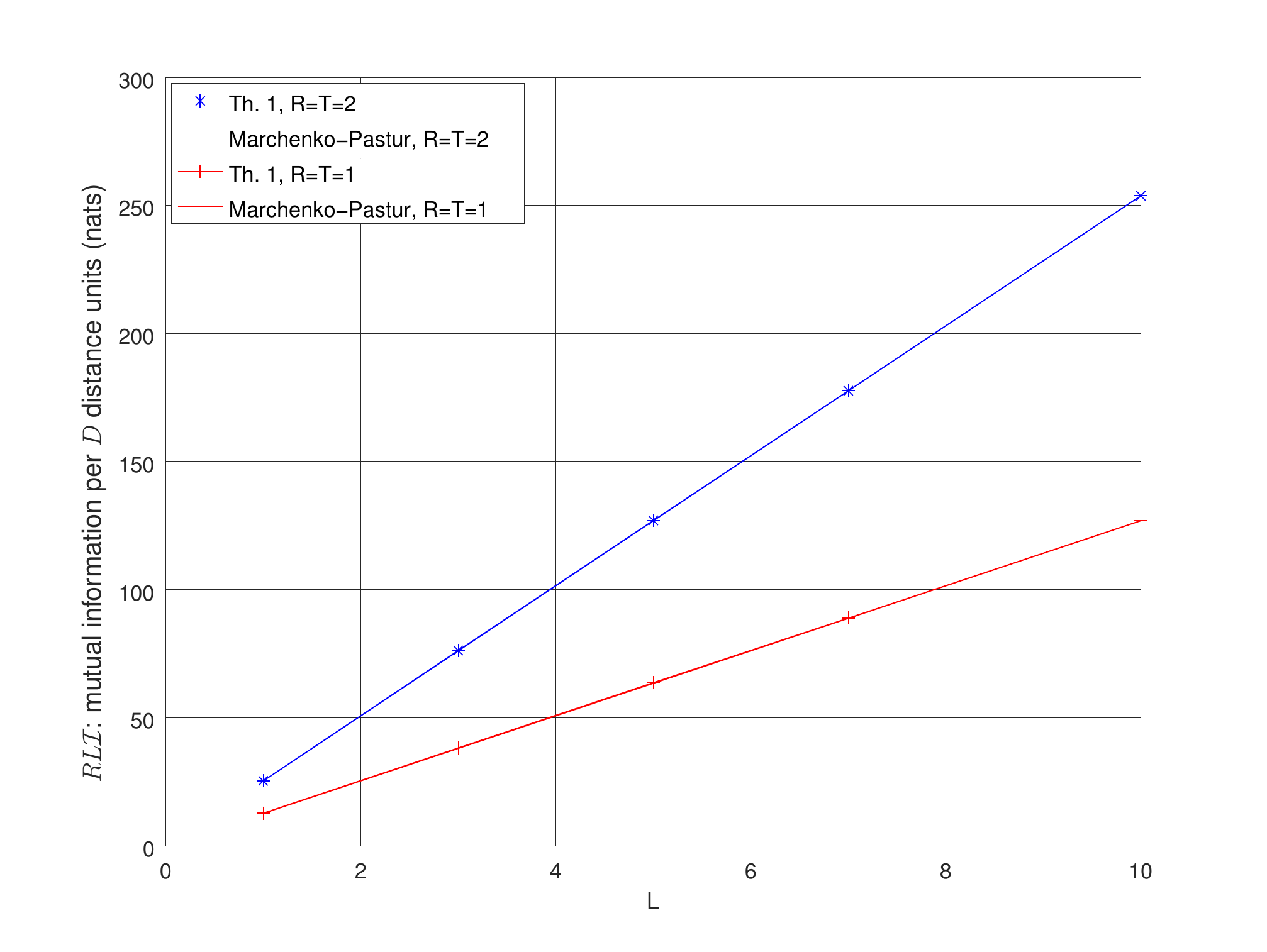}}
\caption{Aggregated mutual information \emph{vs} density of the BS. 
Setting: $\rho = 6 \text{dB}$.} 
\label{fig:marchenko} 
\end{figure}

\paragraph{Illustration of Theorem~\ref{th-Markov}.} Finally, we illustrate the asymptotic behavior of $\mathcal I_\rho$ in the high SNR regime as
predicted by Theorem~\ref{th-Markov}. In this experiment, we consider a more general model than the one described above where we replace the
centered channel coefficient matrix $c_{n,\ell}$ of the model 
by 
\begin{equation} 
\sqrt{\frac{K_{\text{R}}}{K_{\text{R}} + 1}} d_{n,\ell} + 
 \sqrt{\frac{1}{K_{\text{R}} + 1}} c_{n,\ell} , 
\end{equation} 
where $d_{n,\ell} := [ d_{n,\ell}(r,t) ]_{r,t=0}^{R-1, T-1}$ is a determistic
matrix with entries 
\begin{equation} 
d_{n,\ell}(r,t) = a_\ell \exp(2\imath\pi (r-t) \sin(\pi\ell/L) ),
\end{equation}  and
where the nonnegative number $K_{\text{R}}$ plays the role of the 
so-called Rice factor. We take again  
$a_\ell \propto \exp(-0.4\ell)$ and $\alpha = \exp(- f_{\text{d}})$ as in the first paragraph of the section. The high SNR behavior of $\mathcal I_\rho$ is 
illustrated by Figure~\ref{fig:high_snr}.

\begin{figure}[ht]
\centering
{\includegraphics*[scale=0.7]{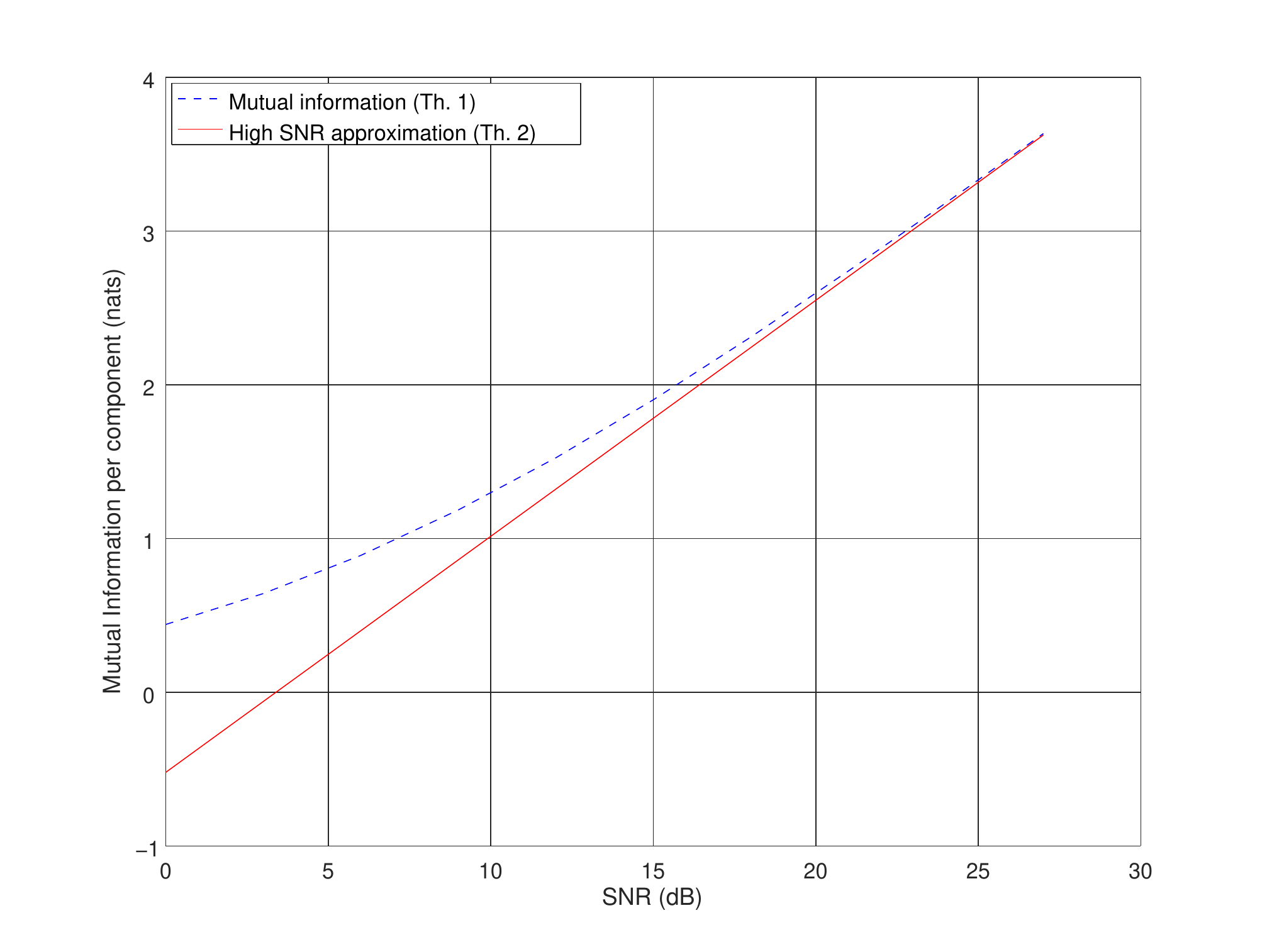}}
\caption{High SNR behavior of $\mathcal I_\rho$. Setting: $R=3$, $T=2$, $L=3$, 
$f_{\text{d}} = 0.05$, $K_{\text{R}} = 10$.}
\label{fig:high_snr} 
\end{figure}
Keeping the same channel model, the behavior of $\kappa_\infty$ in terms of the
Doppler frequency $f_{\text{d}}$ and the Rice factor is illustrated by Figure~\ref{fig:kappa}. 
This figure shows that the impact of $f_{\text{d}}$ is marginal. Regarding
$K_{\text{R}}$, the channel randomness has a beneficial effect on the
mutual information for our model, assuming of course that the 
channel is perfectly known at the receiver. 

\begin{figure}[H]
\centering
{\includegraphics*[scale=0.7]{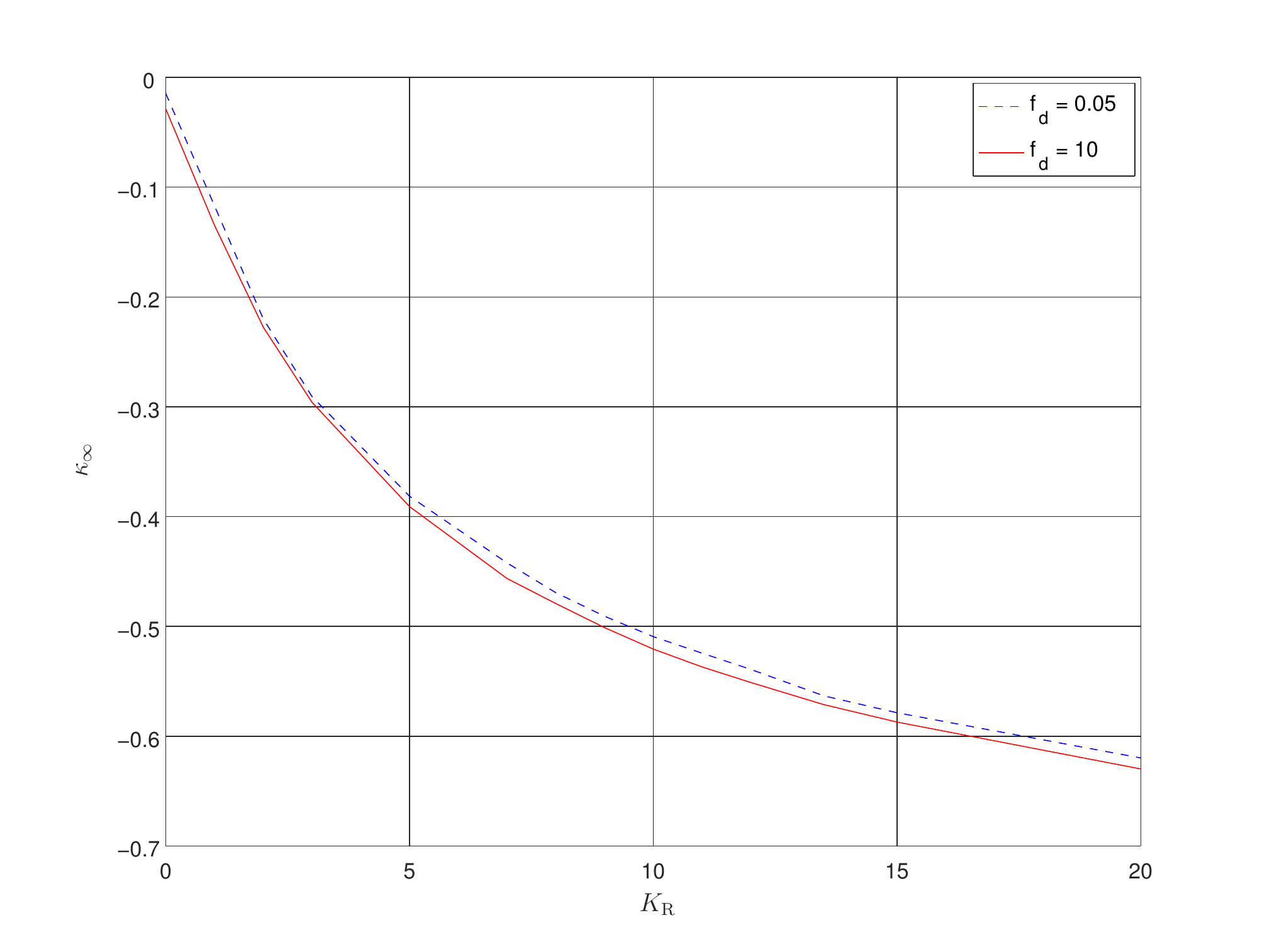}}
\caption{Behavior of $\kappa_\infty$ w.r.t.~$f_{\text{d}}$ and $K_{\text{R}}$. 
Setting: $R=3$, $T=2$, $L=3$.}
\label{fig:kappa} 
\end{figure}

\section{Proofs of Theorem~\ref{th-main}} 
\label{sec:prfmain} 

In this section, we let Assumption~\ref{Ass1} hold true.

\subsection{Preparation}

The idea behind the proof of Theorem~\ref{th-main} is to show that $\mathcal{I}_\rho$ can be given an expression that involves the
resolvents of infinite block-Jacobi matrices and to manipulate these resolvents to obtain the recursion formula for $\W_n$.  We denote for any $m,n\in\ZZ\cup\{\pm\infty\}$ by $\H_{m,n}$ the operator on $\ell^2:=\ell^2(\ZZ)$ defined as the truncation of $\H$, defined in \eqref{def-H}, having the bi-infinite matrix representation
\begin{equation}
\H_{m,n} = \begin{bmatrix} 
\F_{m} & \G_{m} \\ 
      & \F_{m+1} & \G_{m+1}                  \\
      &         & \ddots & \ddots \\ 
      &         &        & \F_{n} & \G_{n} 
\end{bmatrix}  
\end{equation}
where the remaining entries are set to zero. Recalling the definition of the random matrix $\H_{m,n}$ already provided in \eqref{Hmn} for finite $m,n\in\ZZ$, we thus identify this matrix with the associated finite rank operator acting on $\ell^2$ for which we use the same notation.

Let us now introduce a convenient notation: If one considers an operator on $\ell^2$ with block-matrix form
$
\mathsf A = \big[
\mathsf A_{ij}\big]_{i,j\in\ZZ},
$
where the $\mathsf A_{ij}$'s are $Q\times Q$ matrices, then 
$[\mathsf A ]_{\Box_Q}$ stands for the $Q\times Q$ block $\mathsf A_{ii}$ with
largest index $i\in\ZZ$ such that $\mathsf A_{ii}\neq 0$. For the operators of
interest in this work, $[\mathsf A ]_{\Box_Q}$ will always be the bottom
rightmost non-vanishing $Q\times Q$ block. Of importance in the proof will
be the operators of the type $\H_{-\infty,n}$. This operator is closed and
densely defined, thus, defining as $\H_{-\infty,n}^*$ is adjoint, the operator
$\H_{-\infty,n}^* \H_{-\infty,n}$ is a positive self-adjoint 
operator~\cite[Sec.~4]{hmp-jmp15},\cite[Sec.~46]{akh-glaz-93}. Thus, the 
resolvent $(I+\rho \,\H_{-\infty,n}^* \H_{-\infty,n} )^{-1}$ is defined for
each $\rho > 0$, and we can set 
\begin{equation} 
\label{Wn}
\W_n:= [(I+\rho \,\H_{-\infty,n}^* \H_{-\infty,n} )^{-1}]_{\Box_K} . 
\end{equation} 

We shall prove that the sequence $(\W_n)$ indeed satisfies the statements of Theorem~\ref{th-main}. To do so, we will use in a key fashion the following Schur complement identities:
\begin{align}
\label{Schur1}
\det\begin{bmatrix} A & B\\
C& D
\end{bmatrix} & =\det D\times \det(A-BD^{-1}C)\,,\\
\label{Schur2}
\begin{bmatrix} A & B\\
C& D
\end{bmatrix}^{-1}& =\begin{bmatrix} (A-BD^{-1}C)^{-1} & \times \\
\times & (D-CA^{-1}B)^{-1}
\end{bmatrix}\, ,
\end{align}
where the $\times$'s can be made explicit in terms of $A,B,C,D$ but are not of 
interest for our purpose. 

\subsection{Proof of Theorem~\ref{th-main}(a)}
We first show that $\W_n$ defined in \eqref{Wn} indeed satisfies the recursive equations \eqref{recurs}, that is we prove the existence part of Theorem~\ref{th-main}(a).

\subsubsection{Existence}
\begin{proof}[Proof of Theorem~\ref{th-main}(a); existence]Introduce the truncation of $\H_{-\infty,n}$ defined by deleting the rightmost non-zero column,
\begin{equation} 
\widetilde \H_{-\infty,n}:=\begin{bmatrix}  
\ddots  &\ddots\\
         & \F_{n-1} & \G_{n-1} \\ 
         &           &  \F_n
\end{bmatrix} \, , 
\end{equation} 
so that
\begin{equation} 
\H_{-\infty,n}= \left[\begin{array}{c|c} 
 \widetilde \H_{-\infty,n} &   \begin{array}{c} 0 \\ \G_n  \end{array}
\end{array}\right] =:  \left[\begin{array}{c|c} 
 \widetilde \H_{-\infty,n} &  Q
\end{array}\right] .
\end{equation} 
Recalling  $\W_n$'s definition \eqref{Wn}, the Schur's complement formula \eqref{Schur2} then provides

\begin{align}
\W_n &= 
\left(\begin{bmatrix} 
I+\rho\, \widetilde \H_{-\infty,n}^* \widetilde \H_{-\infty,n}  & \rho \, \widetilde \H_{-\infty,n}^*Q \nonumber \\
\rho \,Q^*\widetilde \H_{-\infty,n} &  I+\rho \,\G_n^* \G_n 
\end{bmatrix}^{-1}\right)_{\Box_K} \nonumber \\
 &\stackrel{(a)}{=} \left( I + 
 \rho \,\G_n^* \G_n - \rho^2 \, Q^* \widetilde \H_{-\infty,n}
  (I+\rho \,\widetilde \H_{-\infty,n}^* \widetilde \H_{-\infty,n} )^{-1} 
  \widetilde \H_{-\infty,n}^* Q \right)^{-1} \nonumber \\ 
   &\stackrel{(b)}{=} \left( I + 
 \rho \,\G_n^* \widetilde \W_n\G_n \right)^{-1} 
  \label{Weq} 
\end{align} 
where we introduced 
\begin{equation} 
\widetilde \W_n := 
[(I+\rho\, \widetilde \H_{-\infty,n}^* \widetilde \H_{-\infty,n})^{-1}]_{\Box_N}.
\end{equation} 
Here the identity $\stackrel{(a)}{=}$ can be easily checked similarly to its
finite dimensional counterpart, and $\stackrel{(b)}{=}$ is shown in,
\emph{e.g.}, \cite[Lemma 7.2]{hmp-jmp15}. 

By similarly expressing $\widetilde \H_{-\infty,n}$ in terms of $\H_{-\infty,n-1}$ and $\F_n$, the same computation further yields  
\begin{equation} 
\widetilde \W_n= \left( I +  \rho \,\F_n  \W_{n-1}\F_n^* \right)^{-1}
\end{equation} 
and thus we obtain with \eqref{Weq} the identity
\begin{equation} 
\W_n=\left( I + 
 \rho \,\G_n^* \left( I +  \rho \,\F_n  \W_{n-1}\F_n^* \right)^{-1}\G_n \right)^{-1}.
\end{equation} 
\end{proof}

\subsubsection{Uniqueness}
\label{sec:uniqueness}
Next, we establish the uniqueness of the process $(\W_n)_{n\in\ZZ}$ satisfying the recursive relations~\eqref{recurs} within the class of stationary processes, to complete the proof of Theorem~\ref{th-main}(a). 

The proof relies on a contraction argument with the distance on $\cH_m^{++}$
for $m$ being a positive integer:
\begin{equation} 
\dr : \cH_m^{++} \times \cH_m^{++} \to [0,\infty), \quad
(X,Y) \mapsto \left[ \mathrm{Tr} \log^2(X Y^{-1}) \right]^{1/2} \, ,
\end{equation} 
which is the geodesic distance associated with the Riemannian metric 
$g_X(A,B):=\mathrm{Tr}(X^{-1}AX^{-1}B)$ on the convex cone $\cH_m^{++}$; we 
refer e.g. to \cite[\S1.2]{bou-93} or \cite[\S3]{maa-(livre)71} for further 
information. Convergence in $\dr$ is equivalent to convergence in the 
Euclidean norm. It has the following invariance properties:  for any 
$X, Y \in \cH_m^{++}$ and any  $m\times m$ complex  invertible matrix $A$, 
\begin{equation} 
\label{axa} 
\dr(X,Y) = \dr(AXA^*, AYA^*),\qquad  \dr(X,Y)= \dr(X^{-1}, Y^{-1}) \, . 
\end{equation} 
Moreover, for any $S\in \cH_m^{+}$, we have according to 
\cite[Prop.~1.6]{bou-93},
\begin{equation} 
\label{x+s} 
\dr(X+S, Y+S) \leq 
\frac{\max(\|X\|,\|Y\|)}{\max(\|X\|,\|Y\|)+\lmin(S)} \,\dr(X,Y) \, , 
\end{equation} 
where $\lmin(S)$ is the smallest eigenvalue of $S$. 
We also have the following result, which will be the key to prove the 
uniqueness of the process: 

\begin{lemma}
\label{conj-trans}
Given two positive integers $k$ and $n$ such that $n\geq k$, let 
$X,Y \in \cH_k^{++}$, $S\in \cH_n^{++}$, and $A \in \CC^{n\times k}$. Then, 
\begin{equation} 
\dr(AXA^*+S, AYA^*+S) \leq 
\frac{\max( \| AXA^* \|, \| AYA^* \|)}
 {\max( \| AXA^* \|, \| AYA^* \|)+\lmin(S)} \,\dr(X,Y) \, .
\end{equation}  
\end{lemma}
\begin{proof} 
Define in $\cH_n^{++}$ the two matrices 
\begin{equation}
X' = \begin{bmatrix} X \\ & I_{n-k} \end{bmatrix} 
\quad \text{and} \quad 
Y' = \begin{bmatrix} Y \\ & I_{n-k} \end{bmatrix} . 
\end{equation} 
Let $(B_\ell)$ be a sequence of matrices in $\CC^{n\times n}$ such that 
$B_\ell$ is invertible for each $\ell\in \NN$, and such that 
$B_\ell \to \begin{bmatrix} A & 0 \end{bmatrix}$ as $\ell\to\infty$ (such a 
sequence is guaranteed to exist by the density of the set of invertible 
matrices in $\CC^{n\times n}$). Using the first identity in~\eqref{axa} and 
Inequality~\eqref{x+s}, and observing that $\dr(X,Y) = \dr(X', Y')$, we get 
that 
\begin{equation} 
\dr(B_\ell X' B_\ell^*+S, B_\ell Y' B_\ell^*+S) \leq 
 \frac{\max(\| B_\ell X' B_\ell^*\|, \| B_\ell Y' B_\ell^*\|)}
 {\max(\| B_\ell X' B_\ell^*\|, \| B_\ell Y' B_\ell^*\|) + \lmin(S)} 
 \dr(X, Y) . 
\end{equation} 
Making $\ell\to\infty$, and recalling that the geodesic and the Euclidean 
topologies are equivalent, we obtain the result. 
\end{proof}

\begin{proof}[Proof of Theorem~\ref{th-main}(a); uniqueness] 

To prove the uniqueness, we assume that $N \geq K$ for simplicity, since the
case $N < K$ can be treated in a similar manner.  If one introduces, for any
$F,G\in \CC^{N\times K}$, the mapping $\psi_{F,G}:\cH_K^{++}\to\cH_K^{++}$
defined by
\begin{equation} 
\psi_{F,G}(W):=
\left( I +  \rho \,G^* \left( I +  \rho \,F  W F^* \right)^{-1}G \right)^{-1},
\end{equation} 
then \eqref{recurs} reads  $\W_{n}=\psi_{\F_n,\G_n}(\W_{n-1})$. This mapping 
can be written as 
\begin{equation}
\label{psicomp} 
\psi_{F,G}(W) = \iota \circ \tau_{\sqrt{\rho} G^*; I} 
   \circ \iota \circ \tau_{\sqrt{\rho} F; I}(W) , 
\end{equation} 
where we set 
\begin{equation} 
\label{tauiota} 
\tau_{A;S}(X) :=  AXA^* + S  \quad \text{and} \quad 
\iota(X)  :=  X^{-1} \,  
\end{equation} 
with a small notational abuse related to the fact that, \emph{e.g.}, the two
functions $\iota$ used in~\eqref{psicomp} are not the same in general. Using
Lemma~\ref{conj-trans} together with the invariance of $\dr$ with respect to
the inversion, we obtain  for any $W, W' \in \cH_K^{++}$, 
\begin{align} 
\label{contraction}
& \dr\big(\psi_{F,G}(W), \psi_{F,G}(W') \big) \nonumber\\
&\leq \frac{\rho\| G \|^2}{\rho\| G \|^2 + 1} 
\frac{\max(\| \rho FWF^* \| , \| \rho FW'F^* \|)}
{\max(\| \rho F W F^* \| , \| \rho F W' F^* \|) + 1} \,\dr(W, W') \nonumber \\ 
&\leq \frac{\rho\| G \|^2}{\rho\| G \|^2 + 1} \,\dr(W, W') ,
\end{align} 
where for the first inequality we used that 
$\|G^* (I+\rho F W F^*)^{-1} G\|\leq \| G\|^2$ for any $W\in\cH_K^{+}$, and 
that the function $x\mapsto x / (x+1)$ is increasing.

Now, let $(\W_n')_{n\in\ZZ}$ be any stationary process on $\cH_K^{++}$ satisfying  $\W_{n}'=\psi_{\F_n,\G_n}(\W_{n-1}')$ a.s. for every $n\in\ZZ$. If we let $n \geq 0$, then we have from \eqref{contraction} a.s. that
\begin{equation} 
\dr(\W_{n}, \W'_{n}) \leq 
\frac{\rho \| \G_n \|^2}{\rho \| \G_n \|^2+ 1} \dr(\W_{n-1}, \W'_{n-1}) 
\end{equation} 
and, iterating, we obtain  
\begin{equation} \dr(\W_{n}, \W'_{n}) \leq \Bigl( \prod_{i=1}^n \xi_i \Bigr) \,\dr(\W_{0}, \W'_{0}),\qquad
\xi_i := 
\frac{\rho \| \G_i \|^2}{\rho \| \G_i \|^2 + 1} . 
\end{equation} 
By the ergodicity of $(\G_n)_{n\in\ZZ}$, we have
\begin{equation} 
\frac 1n \sum_{i=1}^n \log \xi_i \xrightarrow[n\to\infty]{\text{a.s.}} 
\EE\log \xi_0<0 
\end{equation} 
and thus we have proven that $\dr(\W_{n}, \W'_{n}) \to 0$ a.s.~as $n\to\infty$.
Finally, since 
\begin{equation} 
(\W_{n+m_1}, \ldots, \W_{n+m_M})
 \stackrel{\text{law}}{=}(\W_{m_1}, \ldots, \W_{m_M})
\end{equation}   
for any $M-$tuple of integers $(m_1,\ldots, m_M)$ and similarly for $\W_n'$, by
letting $n\to\infty$ this yields that the finite-dimensional distributions of
the two stationary processes  $(\W_{n})_{n\in\ZZ}$ and  $(\W'_{n})_{n\in\ZZ}$
are the same, and consequently  these two processes have the same distribution.
\end{proof}

\subsection{Proof of Theorem~\ref{th-main}(b)}
We start with the following lemma. 
\begin{lemma}
\label{le:cvg-resolv} 
For any fixed $n\in\ZZ$ and $\rho>0$, we have
\begin{equation} 
\label{cvg-resolv} 
[ ( I+\rho \,\H_{m,n}^* \H_{m,n})^{-1} ]_{\Box_K} \xrightarrow[m\to-\infty]{} 
\W_n.
\end{equation} 
\end{lemma} 
\begin{proof} 
Denote by $\cK\subset\ell^2$  the subspace of sequences with finite support. Clearly, for any fixed $n\in\ZZ$ and fixed event $\omega\in\Omega$,  we have for all $x\in\cK$,
\begin{equation}  
\H_{m,n}^* \H_{m,n} x \xrightarrow[m\to-\infty]{} 
 \H_{-\infty,n}^* \H_{-\infty,n} x \, , 
\end{equation}  
where $\to$ denotes the strong convergence in $\ell^2$. Now $\cK$ is a common 
core for 
the set of operators $\{  \H_{m,n}^* \H_{m,n} : m \in\{n,n-1,n-2,\ldots\} \}$ 
and $\H_{-\infty,n}^* \H_{-\infty,n}$,  see e.g. \cite[\S III.5.3]{kat-livre80} or \cite[Chap.~VIII]{ree-sim-1}
for this notion. As a consequence, the convergence also holds in the 
strong resolvent sense, see \cite[\S VIII]{ree-sim-1}, and thus for every $x \in \ell^2$ and $\rho>0$,
\begin{equation} 
( I+ \rho \,\H_{m,n}^* \H_{m,n})^{-1} x \xrightarrow[n\to\infty]{} 
( I+\rho \,\H_{-\infty,n}^* \H_{-\infty,n})^{-1} x \quad \text{ a.s.}
\end{equation}  
from which~\eqref{cvg-resolv}  follows by definition \eqref{Wn} of $\W_n$. 
\end{proof}

\begin{proof}[Proof of Theorem~\ref{th-main}(b)]
We start by writing 
\begin{equation} 
\H_{m,n} = \left[\begin{array}{c} 
\begin{array}{c|c} \H_{m,n-1} & 0
 \\ \hline 
  \begin{matrix}  0 & \cdots & 0  & \F_n\end{matrix}   &\G_n
 \end{array} 
\end{array}\right] = 
 \left[\begin{array}{c} 
\begin{array}{c|c} \H_{m,n-1} & 0
 \\ \hline 
 P & \G_n \end{array} 
\end{array}\right]  
\end{equation} 
with $P := [ \ 0 \ \cdots \ 0 \ \F_n ]$, 
and use Schur's complement formula~\eqref{Schur1} to obtain,
\begin{align} 
& \log\det (I+\rho \,\H_{m,n} \H_{m,n}^*) \nonumber \\
&= 
\log\det\begin{bmatrix} 
I+\rho \,\H_{m,n-1} \H_{m,n-1}^*  & \rho \, \H_{n,m-1}^*P^* \nonumber \\
\rho \, P\H_{m,n-1}^* &  I+\rho \,\F_n \F_n^* + \rho \,\G_n \G_n^*
\end{bmatrix} \\
&= \log\det (I+\rho \,\H_{m,n-1} \H_{m,n-1}^*) \nonumber \\
& \phantom{=} + \log\det\left( I + 
\rho \,\F_n \F_n^* + \rho \,\G_n \G_n^* - \rho^2 \, P \H_{m,n-1}^* 
 (I+\rho \,\H_{m,n-1} \H_{m,n-1}^* )^{-1} \H_{n,m-1} P^* \right) \nonumber \\
 &= \log\det (I+\rho \,\H_{m,n-1} \H_{m,n-1}^*) \nonumber \\
& \phantom{=} + \log\det\left( I + 
\rho \,\F_n \F_n^* + \rho \,\G_n \G_n^* + \rho \, P[ 
 (I+\rho \,\H_{m,n-1} ^*\H_{m,n-1} )^{-1}-I]  P^* \right) \nonumber \\
 &= \log\det (I+\rho \,\H_{m,n-1} \H_{m,n-1}^*) \nonumber \\
& \phantom{=} + \log\det\left( I + 
\rho \,\F_n \F_n^* + \rho \,\G_n \G_n^* + \rho \, \F_n [ 
 (I+\rho \,\H_{m,n-1} ^*\H_{m,n-1} )^{-1}-I]_{\Box K} \F_n^* \right) 
                                                             \nonumber \\
 &= \log\det (I+\rho \,\H_{m,n-1} \H_{m,n-1}^*) \nonumber \\
& \phantom{=} + \log\det\left( I + \rho \,\G_n \G_n^* + 
 \rho \,\F_n [(I+\rho \,\H_{m,n-1}^* \H_{m,n-1} )^{-1}]_{\Box_K} \F_n^* \right) .
 \end{align} 
By iterating this manipulation after replacing $\H_{m,n-i}$ by $\H_{m,n-i-1}$ at the $i^{\text{th}}$ step, if we set 
\begin{equation} 
 \xi_{m,i} := \log\det\left( I + \rho\, \G_i \G_i^* + 
 \rho \,\F_i [(I+\rho\, \H_{m,i-1}^* \H_{m,i-1})^{-1}]_{\Box_K} \F_i^* \right) 
\end{equation} 
for any $m\leq i\leq n$ with the convention that $\H_{m,m-1} := 0$, we have
\begin{equation} 
\label{I0nid}
 \log\det\left(I+\rho\, \H_{m,n} \H_{m,n}^* \right) =\sum_{i=m}^{n} \xi_{m,i} \, .
\end{equation}

Next, Lemma~\ref{le:cvg-resolv} yields
\begin{equation} 
\label{asConvXi}
\xi_{m,i} \xrightarrow[m\to -\infty]{}  \log\det\left( I + \rho \,\G_i \G_i^* + 
 \rho \,\F_i \W_{i-1} 
  \F_i^* \right).
\end{equation} 
Since $\| [(I+\rho \,\H_{m,i-1}^* \H_{m,i-1})^{-1}]_{\Box_K} \| \leq 1$, we 
have $\xi_{m,i} \leq N \log(1 + \rho \| \F_i\|^2 + \rho \| \G_i \|^2)$. Thus, 
by the moment assumption~\eqref{momFG}, we obtain from \eqref{asConvXi} and 
dominated convergence that
\begin{align}
\EE \xi_{m,i}  \xrightarrow[m\to-\infty]{} & 
\,\EE\log\det\left( I + \rho \,\G_i \G_i^* + 
 \rho \,\F_i \W_{i-1}
  \F_i^* \right) \nonumber \\ 
  & =\EE\log\det\left( I + \rho \,\G_0 \G_0^* + 
 \rho \,\F_0 \W_{-1}  \F_0^* \right) ,  
\end{align}
where the equality follows from the stationarity of the process 
$(\F_n, \G_n)_{n\in\ZZ}$. The stationarity further provides that 
$\EE \xi_{m,i}$ only depends on $i-m$ and thus, for any fixed $n$, we obtain 
by Ces\`aro summation (see \cite[Page 16]{pol-sze-(livre)98}) that 
\begin{equation} 
\label{prelimComp}
N\mathcal I_\rho=\lim_{m\to-\infty}\frac1{(n-m+1)}\sum_{i=m}^{n} \EE\xi_{m,i}=\EE\log\det\left( I + \rho \,\G_0 \G_0^* + 
 \rho \,\F_0 \W_{-1}  \F_0^* \right).
\end{equation} 
By taking $n=0$ in the recursive relation \eqref{recurs}, we moreover see that 
\begin{align}
\label{relFinal}
N\mathcal I_\rho & =\EE\log\det\left( I + \rho \,\G_0 \G_0^* + 
 \rho \,\F_0 \W_{-1}  \F_0^* \right)\nonumber\\
 & =\EE\log\det\left( I +  \rho \,\F_0 \W_{-1}  \F_0^* \right)+\EE\log\det\left( I + \rho \,\G_0\G_0^* (I+\rho \,\F_0 \W_{-1}  \F_0^*)^{-1} \right)\nonumber\\
 & = \EE\log\det\left( I +  \rho \,\F_0 \W_{-1}  \F_0^* \right)+\EE\log\det\left( I + \rho \,\G_0^* (I+\rho \,\F_0 \W_{-1}  \F_0^*)^{-1}\G_0 \right)\nonumber\\
 & = \EE\log\det\left( I +  \rho \,\F_0 \W_{-1}  \F_0^* \right)-\EE\log\det\W_0,
\end{align}
which proves \eqref{I-WW}.
\end{proof} 

\subsection{Proof of Theorem~\ref{th-main}(c)} \begin{proof}[Proof of Theorem~\ref{th-main}(c)]
Since the process $(\F_n,\G_n)_{n\in\ZZ}$ is assumed to be ergodic, and so does $(\W_n)_{n\in\ZZ}$ by construction, we have a.s.
\begin{multline}
\label{ergoWn}
 \lim_{n\to\infty}\frac1n \sum_{\ell=0}^{n-1}\log\det\left( I +  \rho \,\F_\ell \W_{\ell-1}  \F_\ell^* \right)-\EE\log\det\W_\ell\\
 = \EE\log\det\left( I +  \rho \,\F_0 \W_{-1}  \F_0^* \right)-\EE\log\det\W_0=\mathcal I_\rho.
\end{multline}
Next, for the same reason as and with the same notations as in the proof of the uniqueness of $\W_n$ provided in Section~\ref{sec:uniqueness}, we have $\dr(\X_n, \W_n) \to 0$ a.s. as $n\to\infty$. Thus, 
\begin{equation} 
\frac 1n \sum_{\ell=0}^{n-1} \log\det \X_\ell - \log\det \W_\ell 
\xrightarrow[n\to\infty]{\text{a.s.}} 0 
\end{equation} 
as a Ces\`aro average. Since Lemma~\ref{conj-trans} also yields
\begin{equation} 
\dr(I+\rho\,\F_n\X_{n-1}\F_n \,,\,I+\rho\,\F_n\W_{n-1}\F_n )\leq \dr(\X_{n-1},\W_{n-1})\xrightarrow[n\to\infty]{\text{a.s.}} 0, 
\end{equation} 
we similarly have
\begin{equation} 
\frac 1n \sum_{\ell=0}^{n-1} \log\det (I+\rho\,\F_\ell\X_{\ell-1}\F_\ell^*) - \log\det(I+\rho\,\F_\ell \W_{\ell-1}\F_\ell^*) 
\xrightarrow[n\to\infty]{\text{a.s.}} 0 .
\end{equation} 
and the result follows from this convergence along with \eqref{ergoWn}. 
\end{proof}

This completes the proof of Theorem~\ref{th-main}.

\section{Proof of Theorem~\ref{th-Markov}} 
\label{sec:prfasymp} 

Assume from now that $N> K$ and that Assumption~\ref{Ass2} holds true.

\subsection{Preparation}
To obtain an expansion of the type 
$\mathcal I_\rho = (K/N) \log\rho + \kappa_\infty + o(1)$ as $\rho\to\infty$, 
it is more convenient to work with the new variables:
\begin{equation} 
\label{defWn}
\gamma:=\frac1\rho \in(0,\infty)\,,\qquad \Zg{n}:=\gamma\, \W_n^{-1}\,.
\end{equation} 
Indeed, it follows the identity \eqref{I-WW} of Theorem~\ref{th-main} and the stationarity of $(\W_n)_{n\in\ZZ}$ that
\begin{align}
N \mathcal I_\rho &= -\EE \log\det \W_0  +\EE\log\det(I+\rho\,\F_{1}\W_n\F_{1}^*)
  \nonumber \\
&= K \log\rho + \EE\log\det \Zg{0} + 
   \EE\log\det (I + \F_{1} \Zg{0}^{-1} \F_{1}^*) \nonumber \\ 
   &= K \log\rho + \EE\log\det \Zg{0} + 
   \EE\log\det (I + \Zg{0}^{-1} \F_{1}^*\F_{1} ) \nonumber \\ 
&= K \log\rho + \EE\log\det (\Zg{0} + \F_{1}^* \F_{1})  ,
\label{Ihighsnr} 
\end{align} 
which is the starting point of the asymptotic analysis $\gamma\to0$.
With this expression at hand, we would like to take the limit  $\gamma\to 0$ 
and identify the limit
\begin{equation} 
\label{kappaIdea}
\kappa_\infty:=\frac1N\lim_{\gamma\to 0} \EE\log\det (\Zg{0} + \F_{1}^* \F_{1}).
\end{equation}

To study this limiting case, we start from the recursive 
equation~\eqref{recurs}, which reads for these new variables 
\begin{equation} 
\Zg{n} 
=  \gamma I + \G_n^* ( I + \F_n \Zg{n-1}^{-1} \F_n^*)^{-1} \G_n = h_{\gamma,\F_n,\G_n}(\Zg{n-1}) ,
\end{equation} 
where, for any $\gamma\geq 0$ and $F,G\in \CC^{N\times K}$, we define 
 $h_{\gamma,F,G} : \cH_K^{++} \to \cH_K^{+}$ 
 by
\begin{equation}
\label{defh} 
h_{\gamma,F,G}(Z) := \gamma I+G^* ( I + F Z^{-1} F^*)^{-1} G \, . 
\end{equation} 
Note that if $\gamma>0$ then $h_{\gamma,F,G}(Z)\in\cH_K^{++}$. The same holds true when $\gamma=0$, which is now allowed, as soon as $G$ has full rank. We now observe that one can extend this mapping to the whole of $\cH_K^+$. 

\subsubsection{Extension of the mapping $h_{\gamma,F,G}$ to $\cH_K^+$}

Assume that $F\in\CC^{N\times K}$ has full rank, namely $\rank(F) = K$. By setting $T:=(F^*F)^{1/2}$ and $U:=F(F^*F)^{-1/2}$,  we have the polar decomposition $F = U T$ where $U\in\CC^{N\times K}$ is an
isometry matrix and $T\in\cH_K^{++}$. By completing $U$
so as to obtain a $N\times N$ unitary matrix $\begin{bmatrix} U \ \  U^\perp
\end{bmatrix}$ and setting $\Pi_{F}^\perp:=U^\perp(U^\perp)^*=I-F(F^*F)^{-1}F^*$, which the orthogonal projection onto the orthogonal space to the linear span of the columns of $F$, we can write 
\begin{align} 
& h_{\gamma, F, G}(Z) \nonumber \\
& =\gamma I+G^* ( I + F Z^{-1} F^*)^{-1} G  \nonumber\\
&= \gamma I+G^*U ( I+T Z^{-1} T)^{-1} U^* G + G^* \Pi_{F}^\perp G 
  \label{hsemiUT}\\
&=  \gamma I + G^*U T^{-1} Z^{1/2} 
     ( I + Z^{1/2} T^{-2} Z^{1/2})^{-1} Z^{1/2} T^{-1} U^* G 
    + G^* \Pi_{F}^\perp G \nonumber\\
&= \gamma I+ G^*F(F^*F)^{-1} Z^{1/2}(I+ Z^{1/2} (F^*F)^{-1} Z^{1/2})^{-1} 
     Z^{1/2}(F^*F)^{-1}F^*G+G^* \Pi_{F}^\perp G \label{hsemi}  
\end{align} 
where for the second equality we used the matrix identity $(I+AB)^{-1}=B^{-1}(I+A^{-1}B^{-1})^{-1}A^{-1}$ with $A:=TZ^{-1/2}$ and $B:=Z^{-1/2}T$ for any $Z^{1/2}\in\cH_K^+$ satisfying $(Z^{1/2})^2=Z$. Note that the  alternative expression \eqref{hsemi} for $h_{\gamma, F, G}(Z)$  does  now make sense when $Z\in\cH_K^+$ is not invertible, provided that $F$ has full rank. Moreover, since two Hermitian square roots of $Z\in\cH_K^+$ are identical up to the multiplication by a unitary matrix, the right hand side of \eqref{hsemi} does not depend on the choice for $Z^{1/2}$. In the following, we chose $Z\mapsto Z^{1/2}$ so that it is continuous (for the operator norm). Thus, by taking the right hand side of \eqref{hsemi} as the definition of $h_{\gamma, F, G}(Z)$ in this case, we properly extended $h_{\gamma, F, G}$ to a mapping $\cH_K^+\to\cH_K^+$ which is continuous, and that we continue to denote by $h_{\gamma, F, G}$. An important property of $h_{0,F,G}$ we use in what follows is:

\begin{lemma}
\label{hinc} 
If $F$ has full rank, then $h_{0,F,G}:\cH_K^+\to\cH_K^+$ is non-decreasing. 
\end{lemma}
\begin{proof}
It is clear from \eqref{defh} this mapping is non-decreasing on $\cH_K^{++}$ and this property extends to $\cH_K^{+}$ since one can write $h_{0,F,G}(Z)=\lim_{\varepsilon \to 0}h_{0,F,G}(Z+\varepsilon I)$ by continuity of $h_{0,F,G}$. 
\end{proof}

\subsubsection{The Markov kernel $Q_\gamma$}
\label{sec:Qgamma}
Equipped with the extended definition of $h_{\gamma, F, G}$ to $\cH_K^+$, let us consider for any $\gamma\geq 0$ the Markov transition kernel $Q_\gamma : (E\times\cH_K^{+}) \times \mcB(E\times \cH_K^{+}  ) \to [0,1]$  defined by 
\begin{equation} 
Q_\gamma f(\F,\G,\Z):=\int f\big(F,G, h_{\gamma,F,G}(\Z)\big) P\big((\F,\G), dF\times dG\big)
\end{equation} 
for any $(\F,\G)\in E$, any $\Z\in\cH_K^+$ and any Borel test function $f:E\times\cH_K^+\to[0,\infty)$. 

\begin{remark} 
In the following, we will use at several instances the following fact: Since
$\theta =\theta P$, Assumption~\ref{Ass2}(d) yields that $G^*F$ is non-singular,
and thus that both $F$ and $G$ have full rank, $P((\F,\G),\cdot)$-a.s. for
$\theta$-a.e. $(\F,\G)$. In particular, $Q_0 f(\F,\G,\Z)$ is properly defined
for $\theta$-a.e. $(\F,\G)$, which will be enough for our purpose.
\end{remark}

When $\gamma>0$, if $(\F_n,\G_n,\Zg{n})_{n\in\ZZ}$ denotes the Markov process
defined by $\Zg{n}=h_{\gamma,\F_n,\G_n}(\Zg{n-1})$ with $(\F_n,\G_n)_{n\in\ZZ}$
the Markov process with transition kernel $P$, then by the definition of
$\Zg{n}$ in~\eqref{defWn} and by Theorem~\ref{th-main}, it follows that
$Q_\gamma$ has a unique invariant measure, that we denote by $\pi_\gamma$. The
strategy of the proof of Theorem~\ref{th-Markov} is to show that $Q_0$ has
also a unique invariant measure $\pi_0$, which will yield the existence of the
process $\Z_n:=\Z_{0,n}$, and we also show that $\pi_\gamma\to\pi_0$ narrowly as
$\gamma\to 0$ and that one can legally take the limit $\gamma\to0$ in
\eqref{kappaIdea}, so as to obtain $N\mathcal I_\rho + K \log \gamma
\to\EE\det(\Z_0+\F_1^*\F_1)$. It turns out when $N=K$ one can possibly lose the
uniqueness of the invariant measure for $Q_0$, which makes this setting out of
reach for our current approach.

\subsection{Existence and uniqueness of the invariant measure of $Q_0$}

The key to prove the existence of an invariant measure for $Q_0$ is the following result.
\begin{lemma} 
\label{tightness} 
The family of probability measures  on $\cH_K^+$,
\begin{equation} 
\label{tightset}
\mathscr C:=\big\{ \zeta Q_0^n(E\times \cdot )  
  \, : \; \zeta \in \cM(E\times\cH_K^+) ,\;
\zeta(   \cdot\times \cH_K^+) = \theta(\cdot), \ n \geq K \big\}.
\end{equation} 
 is a tight subset of  $\cM(\cH_K^{++}).$
\end{lemma} 

\begin{proof} Let us fix $\varepsilon>0$. We first prove there exists $\eta>0$ such that, for any $\xi\in\mathscr C$, 
\begin{equation} 
\label{tight++}
\xi(\lambda_{\rm{min}}(Z)\geq \eta)\geq 1-\varepsilon,
\end{equation} 
where we recall that $\lmin(Z)$ is the smallest eigenvalue of $Z\in\cH_K^+$. To do so, observe from \eqref{defh} that if $Z\in\cH_K^{++}$ then so does $h_{0,F,G}(Z)$ as soon as $G$ has full rank, which is true $\theta$-a.s. due to Assumption~\ref{Ass2}(d). We claim that this assumption further yields that, that for all $(\F,\G,\Z)$ satisfying $\rank(\Z) < K$, we have
$Q_0((\F,\G,\Z), \;\rank(Z) > \rank(\Z)) = 1$, 
namely  at each step of the process the rank of the random matrix $Z$ increases $Q_0((\F,\G,\Z),\cdot)$-a.s.  To prove this, we start from 
\begin{equation} 
Q_0((\F,\G,\Z), \rank(Z) \leq \rank(\Z)) = P((\F,\G), \;\rank (h_{0,F,G}(\Z))\leq \rank(\Z)).
\end{equation} 
Recalling \eqref{hsemi}, we have 
$\rank(h_{0,F,G}(\Z)-G^* \Pi_{F}^\perp G)=\rank(\Z)$ as soon as $F^*G$ is
invertible. Using Assumption~\ref{Ass2}(d) in conjunction with the general
fact that $\rank(A+B)\leq \rank(A)$ implies that
the column spans of these matrices satisfy $\colspan(B)\subset\colspan(A)$ for
any $A,B\in\cH_K^+$, this yields  
\begin{equation} 
Q_0((\F,\G,\Z), \rank(Z) \leq \rank(\Z)) =P\big((\F,\G), \;\colspan(G^* \Pi_{F}^\perp G)\subset\colspan(h_{0,F,G}(\Z)-G^* \Pi_{F}^\perp G)\big)
\end{equation} 
for $\theta$-a.e. $(\F,\G)$. Next, we will use repeatedly that,  for two matrices $A$ and $B$ we have $\colspan(A) 
\subset\colspan(B)$ if and only if 
$\colspan(CAD) \subset \colspan(CBD)$ for all invertible matrices $C$ and $D$. If we let  $\Z^\perp \in \CC^{K\times K}$ be
any matrix such that $\colspan(\Z^\perp)=\colspan(\Z)^\perp$, we have:
\begin{align}
&\colspan(G^* \Pi_{F}^\perp G)\subset
 \colspan(h_{0,F,G}(\Z)-G^* \Pi_{F}^\perp G) \nonumber \\
\Leftrightarrow&  \colspan(G^* \Pi_{F}^\perp G)\subset
  \colspan(G^*F(F^*F)^{-1}\Z^{1/2}(I+\Z^{1/2} (F^*F)^{-1}\Z^{1/2})^{-1}
  \Z^{1/2}(F^*F)^{-1}F^*G) \nonumber \\
\Leftrightarrow&  
\colspan(G^* G  - G^* F (F^* F)^{-1} F^* G) \nonumber \\
& \qquad\qquad \qquad\;\,\subset \colspan(G^*F(F^*F)^{-1}\Z^{1/2}(I+\Z^{1/2} 
  (F^*F)^{-1}\Z^{1/2})^{-1}\Z^{1/2}(F^*F)^{-1}F^*G) \nonumber \\
\Leftrightarrow& 
  \colspan( F^* F (G^* F)^{-1} G^* G (F^* G)^{-1} F^* F - F^* F ) 
  \subset \colspan( \Z) \nonumber \\
\Leftrightarrow& \colspan( F^* F ( F^* \Pi_G F )^{-1} F^* F - F^* F ) 
  \subset \colspan( \Z)  \nonumber \\ 
\Leftrightarrow& F^* F ( F^* \Pi_G F )^{-1} F^* F \Z^\perp - F^* F \Z^\perp 
    = 0 \nonumber \\
\Leftrightarrow&  F^* F \Z^\perp =  F^* \Pi_G F \Z^\perp \nonumber \\
\Leftrightarrow&  F^* \Pi_G^\perp F \Z^\perp  = 0, \, 
\end{align} 
provided that  $F$ and $ G$ have full rank. 
Therefore, together with Assumption~\ref{Ass2}(d), we obtain
\begin{equation} 
Q_0((\F,\G,\Z), \rank(Z) \leq \rank(\Z)) 
 = P((\F,\G),\; F^* \Pi_G^\perp F \Z^\perp  = 0 ) = 0 ,
\end{equation} 
for $\theta$-a.e. $(\F,\G)$, and our claim follows. 
As a consequence, $Z$ has full rank $(\theta\otimes \delta_0 )Q_0^K((\F,\G),\cdot)$-a.s. and thus there exists $\eta>0$ such that 
\begin{equation} 
(\theta\otimes \delta_0 )Q_0^K((\F,\G), \; \lmin(Z)\geq \eta)\geq 1-\varepsilon.
\end{equation} 

Next, we use that $Z\mapsto h_{0,F,G}(Z)$ and  $Z\mapsto \lmin(Z)$ are 
non-decreasing on $\cH_K^+$, see Lemma~\ref{hinc}, so that for any 
$\zeta \in \cM(E\times\cH_K^+)$ satisfying 
$\zeta( \cdot\times \cH_K^+) = \theta(\cdot)$ and any $n\geq K$, we have
\begin{align}
\zeta Q_0^n( \lmin(Z)\geq \eta)& 
 \geq (\theta\otimes\delta_0)Q_0^n( \lmin(Z)\geq \eta) \nonumber \\
&= \big((\theta\otimes\delta_0)Q_0^{n-K}(E\times\cdot)\big) 
 Q_0^K( \lmin(Z)\geq \eta) \nonumber \\
& \geq Q_0^K( \lmin(Z)\geq \eta)\geq 1-\varepsilon,
\end{align}
which finally proves \eqref{tight++}.

Finally, let $C>0$ be such that $\theta( \| G\|^2 > C )< \varepsilon$ and consider the compact subset $\mathcal K$ of $\cH_K^{++}$ given by 
\begin{equation} 
\mathcal K:=\big\{Z\in\cH_K^{++}:\; \lmin(Z)\geq \eta,\quad \|Z\|\leq C\big\}.
\end{equation} 
 It follows from \eqref{hsemi} that $\| h_{0, F, G}(Z) \| \leq \| G \|^2$ for any $(F,G)\in E$ such that $F$ has full rank and any  $Z\in\cH_K^{+}$. This provides, for any $\zeta\in\cM(E\times\cH_K^+)$ satisfying $\zeta( \cdot\times \cH_K^+) = \theta(\cdot)$ and any  $n \geq K, $ 
\begin{align}
\zeta Q_0^n( \| Z \| > C ) & \leq \zeta Q_0^n( \| G\|^2 > C ) \nonumber \\
& = \theta P^n( \| G\|^2 > C ) \nonumber \\
& = \theta( \| G\|^2 > C )
 < \varepsilon \, 
\end{align}
and thus $\xi(\mathcal K)\geq 1-2\varepsilon$ for any $\xi\in\mathscr C$. The proof of the lemma is therefore complete.
\end{proof}

In the remainder, $\Cb(S)$ denotes the set of continuous and bounded functions
on the metric space $S$. 

\begin{lemma}  
\label{le:Feller} For any $\gamma \geq 0$ the kernel $Q_\gamma$ maps 
$\Cb(E\times \cH_K^{++})$ to itself. 
\end{lemma}

\begin{proof} Let $f : E\times \cH_K^{++} \to\RR$ be a  bounded and continuous function, and note from the definition of $Q_\gamma$ that $Q_\gamma f$ is clearly bounded. To show it is continuous, let  $(\F_k,\G_k,\Z_k)_{k\geq 1}$  be a sequence converging to $ (\F_0,\G_0,\Z_0)$ in $E\times \cH_K^{++}$ as $k\to\infty$. If we set $g_k(F,G):=f(F,G,h_{\gamma,F,G}(\Z_k))$ and $\mu_k(\cdot):=P((\F_k,\G_k),\cdot)$, then this amounts to show that $\int g_k\, d\mu_k\to\int g_0\, d\mu_0$ as $k\to\infty$. Since $P$ is Feller by Assumption~\ref{Ass2}(a), we have the narrow convergence $\mu_k\to\mu_0$ . Since $(F,G)\mapsto h_{\gamma,F,G}(Z)$ is continuous  on $E$ for any $Z\in\cH_K^{++}$ we have $g_0\in\Cb(\cH_K^{++})$ and that  $g_k\to g_0$ locally uniformly on $E$. Together with the tightness of $(\mu_k)$ and that $\sup_{k\in\NN}\|g_k\|_\infty<\infty$, we obtain $\int g_k\, d\mu_k\to\int g_0\, d\mu_0$ and the proof of the lemma is complete.
\end{proof}

\begin{corollary} 
\label{cor:existQ0} 
$Q_0$ has an invariant measure in $\cM(E\times\cH_K^{++})$. 
\end{corollary}
\begin{proof}
Let $\zeta :=\theta\otimes\delta_0$ so that by Lemma~\ref{tightness} we have $\zeta Q_0^n\in\cM(E\times \cH_K^{++})$ for every $n\geq K$ and $\zeta Q_0^n\to \pi$ narrowly as $n\to\infty$  for some $\pi\in\cM(E\times\cH_K^{++})$, possibly up to the extraction of a subsequence. If we set, for any $n>K$,
\begin{equation} 
\zeta \bar Q_{0,n}:= \frac{1}{n-K} \sum_{\ell=K}^{n-1} \zeta Q_0^\ell\, \in \cM(E\times \cH_K^{++}),
\end{equation} 
then we also have the narrow convergence $\zeta \bar Q_{0,n}\to \pi$. Next, given any $f\in\Cb(E\times \cH_K^{++})$, we write
\begin{equation} 
\zeta \bar Q_{0,n}f= \frac{\zeta Q_0^Kf}{n-K} + \zeta \bar Q_{0,n}(Q_0 f) 
  - \frac{\zeta Q_0^n f}{n-K} . 
\end{equation} 
Since $Q_0f\in\Cb(E\times\cH_K^{++})$ according to Lemma~\ref{le:Feller}, by taking the limit $n\to\infty$ we obtain $\pi f=\pi Q_0f$ and thus $\pi$ is an invariant measure for $Q_0$.
\end{proof}

\begin{lemma} 
\label{le:uniqueQ0}
If $Q_0$ has an invariant distribution, then it is unique.
\end{lemma}

\begin{proof} 
If $\pi\in\cM(E\times \cH_K^+)$ satisfies $\pi=\pi Q_0$ then $\pi=\pi Q_0^K$
and Lemma~\ref{tightness} yields that necessarily $\pi\in\cM(E\times
\cH_K^{++})$.  Let $\pi^1,\pi^2\in\cM(E\times \cH_K^{++})$ be two invariant
distributions for $Q_0$.  Since $\theta$ is the unique invariant distribution for $P$ by
assumption, necessarily $\pi^i(\cdot\times \cH_K^{++})=\theta(\cdot)$. 
Let $\X_0^{\pi^1} := (\F_0, \G_0, \Z_0^{\pi^1})$ and 
$\X_0^{\pi^2} := (\F_0, \G_0, \Z_0^{\pi^2})$ be two 
$E\times \cH_K^{++}$--valued random variables such that 
$\X_0^{\pi^i} \sim \pi^i$. Starting from $\X_0^{\pi^1}$ and $\X_0^{\pi^2}$,  
construct two Markov processes 
$(\X_n^{\pi^i} = ( \F_n, \G_n, \Z_n^{\pi^i}))_{n\in\NN}$ with the transition
kernel $Q_0$ for $i=1,2$ respectively. To show that $\pi^1 = \pi^2$, it will 
be enough to show that $\| \X_n^{\pi^1}  - \X_n^{\pi^1} \| \to 0$ in 
probability as $n\to\infty$, or equivalently, that 
$\dr(\Z_n^{\pi^1}, \Z_n^{\pi^2}) \to 0$ in probability. We use similar 
arguments and the same notations as in Section~\ref{sec:uniqueness}.

Recalling \eqref{hsemiUT} for $\gamma = 0$, and keeping in mind
that Assumption~\ref{Ass2}(d) yields that $\Z_n\in\cH_K^{++}$ a.s. and that
$\F_n$ has full rank a.s. for every $n\in\NN$, we have
\begin{equation} 
\Z_n^{\pi^i}=h_{0,\F_{n},\G_{n}}(\Z_{n-1}^{\pi^i}) 
= \tau_{\G_n^* \F_n (\F_n^* \F_n)^{-1/2},\, \G_n^* \Pi_{\F_n}^\perp \G_n} 
 \circ \iota 
 \circ \tau_{(\F_n^* \F_n)^{1/2}, I} \circ \iota (\Z_{n-1}^{\pi^i}) . 
\end{equation} 
Dealing with the terms $\tau_{(\F_n^* \F_n)^{1/2}, I}$ and 
$\tau_{\G_n^* \F_n (\F_n^* \F_n)^{-1/2},\, \G_n^* \Pi_{\F_n}^\perp \G_n}$ 
by Lemma~\ref{conj-trans} and Inequality~\eqref{x+s} respectively, we 
get 
\begin{equation} 
\dr( \Z_n^{\pi^1}, \Z_n^{\pi^2} ) \leq 
\frac{\| \F_n\|^2  
  \max(\| (\Z_{n-1}^{\pi^1})^{-1} \| , \| (\Z_{n-1}^{\pi^2})^{-1} \| )}
{\| \F_n\|^2 \max(\| (\Z_{n-1}^{\pi^1})^{-1} \| , 
  \| (\Z_{n-1}^{\pi^2})^{-1} \| ) + 1} \dr(\Z_{n-1}^{\pi^1},\Z_{n-1}^{\pi^2}) ,
\end{equation} 
which implies that, for any $n\geq 1$,  
\begin{equation} 
\label{boundDR0}
\dr(\Z_{n}^{\pi^1},\Z_{n}^{\pi^2}) \leq 
\Bigl( \prod_{i={0}}^{n-1} \xi_i \Bigr)
\dr(\Z_{0}^{\pi^1},\Z_{0}^{\pi^2}),
 \quad 
\xi_i := \frac{\| \F_{i+1}\|^2  \max(\| (\Z_{i}^{\pi^1})^{-1} \| , \| (\Z_{i}^{\pi^2})^{-1} \| )}
{\| \F_{i+1}\|^2 \max(\| (\Z_{i}^{\pi^1})^{-1} \| , \| (\Z_{i}^{\pi^2})^{-1} \| ) + 1} . 
\end{equation} 
By H\"older's inequality, we have 
\begin{equation} 
\EE \prod_{i=0}^{n-1} \xi_i \leq 
\prod_{i=0}^{n-1} (\EE\xi_i^n)^{1/n} =  
\EE\left[\left( 
\frac{\| \F_{1}\|^2  \max(\| (\Z_{0}^{\pi^1})^{-1} \| , \|
(\Z_{0}^{\pi^2})^{-1} \| )} {\| \F_{1}\|^2 \max(\| (\Z_{0}^{\pi^1})^{-1} \| ,
\| (\Z_{0}^{\pi^2})^{-1} \| ) + 1} 
 \right)^n\right]  . 
\end{equation} 
By dominated convergence, the rightmost term of these inequalities converges
to zero as $n\to\infty$, and thus $ \prod_{i=0}^{n-1} \xi_i \to 0$ in probability. It thus follows from \eqref{boundDR0} that 
$\dr( \Z_n^{\pi^1}, \Z_n^{\pi^2} ) \to 0$ in probability, which concludes the
proof. 
\end{proof}

\subsection{The last step for the proof of Theorem~\ref{th-Markov}} 

\begin{proof}[Proof of Theorem~\ref{th-Markov}] 
First, Corollary~\ref{cor:existQ0} and Lemma~\ref{le:uniqueQ0} show that $Q_0$
has a unique invariant measure, that we denote by $\pi_0$, and moreover that
$\pi_0\in\cM(E\times\cH_K^{++})$. Kolomorogov's existence theorem then yields
there exists a unique stationary Markov process $(\F_n,\G_n,\Z_n)_{n\in\ZZ}$ on
$E\times\cH_K^{++}$ with transition kernel $Q_0$, which is in particular
ergodic. Moreover, $(\Z_n)_{n\in\ZZ}$ satisfies the equation~\eqref{eq:Zn} by
definition of $Q_0$, which proves part (a) of the theorem. 

To prove (b), we claim that the family $\{ \pi_\gamma \}_{\gamma\in [0, 1]}$ is tight in 
$\cM(E\times \cH_K^+)$. Indeed, if $(F,G,Z) \sim \pi_\gamma$, then $Z=h_{\gamma,F,G}(Z)$ in law  and, since $\|h_{\gamma,F,G}(Z)\|\leq \|G\|^2+\gamma$ and $\pi_\gamma(\cdot\times\cH_K^{++})=\theta(\cdot)$ is independent on $\gamma$, the claim follows. As a consequence, $\pi_\gamma \to\zeta$ narrowly  for some $\zeta\in\cM(E\times\cH_K^{+})$ as $\gamma\to 0$ along a subsequence. By definition of $\pi_\gamma$, for any $f \in \Cb(E\times\cH_K^{++})$ we  have
\begin{equation} 
\pi_{\gamma}f =\pi_{\gamma}Q_{\gamma} f. 
\end{equation} 
The left hand side converges to $\zeta f $ as $\gamma\to 0$ by definition of $\zeta$, and the exact same lines of arguments as in the proof of Lemma~\ref{le:Feller} yield that the right hand side converges to $\zeta Q_{0} f$, showing that $\zeta=\zeta Q_0$. Since the invariant measure $\pi_0$ of $Q_0$ is unique, we thus have shown that $\pi_\gamma\to\pi_0$ narrowly as $\gamma\to 0$.

We finally go back to the identity~\eqref{Ihighsnr}, which can be rewritten as 
\begin{align}
\label{KeyI}
N{\mathcal I_\rho} + K\log\gamma & = \EE\log\det(\Z_{\gamma,0}+\F_1^*\F_1)\nonumber\\
& =  \int \log\det(\Z+F^*F)P((\F,\G),dF\times dG)\pi_\gamma(d\F\times d\G\times d\Z)\nonumber\\
& =  \int \log\det(\Z+F^*F)Q_\gamma((\F,\G,\Z),dF\times dG\times dZ)\pi_\gamma(d\F\times d\G\times d\Z)\nonumber\\
& =  \int \log\det(\Z+\F^*\F)\pi_\gamma(\F,\G,\Z)\nonumber\\
& = \EE\log\det(\Z_{\gamma,1}+\F_1^*\F_1).
\end{align}
Using Skorokhod’s representation theorem, we can introduce a probability space
$(\Omega', \mcF', \PP')$ and a family of $E \times \cH_K^{++}$-valued 
random variables 
$\{\U_{\gamma,i}=( \F_i', \G_i', \Z'_{\gamma,i}) 
 \, : \, \gamma\in [0,1],\; i\in\{0,1\} \}$  such that $\U_{\gamma,i} \sim \pi_\gamma$, $\Z'_{\gamma,1}=h_{\F_1',\G_1',\gamma}(\Z'_{\gamma,0})$ for every $\gamma\in[0,1]$ and  
 $\U_{\gamma,i} \to \U_{0,i}$  as $\gamma\to 0$ $\PP'$-a.s. This yields that,  as $\gamma\to 0$,
 \begin{equation} 
 \log\det(\Z'_{\gamma,1}+(\F'_1)^*\F'_1) \to \log\det(\Z'_{1,0}+(\F'_1)^* \F'_1),\qquad \PP'\text{-a.s.}
 \end{equation} 
 Moreover, using that 
$\|\Z_{\gamma,1}' \|=\|h_{\F_1',\G_1',\gamma}(\Z_{\gamma,0}') \| \leq \gamma + \| \G'_1  \|^2 
\leq 1 + \| \G'_1  \|^2$ and that $\log(1 + a + b) \leq \log(1+a) + \log(1+b)$ for any $a,b\geq 0$, we
also have 
\begin{equation} 
\log\det((\F'_1)^*\F'_1)\leq \log\det(\Z'_{\gamma,1}+(\F'_1)^*\F'_1) \leq   N \log( 1 + \| \G'_1 \|^2 + \| \F'_1 \|^2)
\end{equation} 
and thus 
\begin{equation} 
|\log\det(\Z'_{\gamma,1}+(\F'_1)^*\F'_1)|\leq  H(\F'_1,\G'_1) := | \log\det( (\F'_1)^* \F'_1) | 
      + N \log( 1 + \| \G'_1 \|^2) 
      + N \log( 1 + \| \F'_1 \|^2) . 
\end{equation} 
Since $(\F_1',\G_1')$ has law $\theta$ by construction, Assumption~\ref{Ass2}(b)-(c) yields that $\EE H(\F'_1,\G'_1) <\infty$ and thus, by dominated convergence, we obtain from \eqref{KeyI},
\begin{align}
\lim_{\gamma\to 0}N{\mathcal I_\rho} + K\log\gamma &=
 \lim_{\gamma\to 0} \EE\log\det(\Z_{\gamma,1}+\F_1^*\F_1) \nonumber \\
& =\lim_{\gamma\to 0} \EE\log\det(\Z'_{\gamma,1}+(\F'_1)^*\F'_1) \nonumber \\
&= \EE\log\det(\Z'_{0,1}+(\F'_1)^*\F'_1) \nonumber \\
&= \EE\log\det(\Z'_{0,0}+(\F'_1)^*\F'_1) \nonumber \\
&= \EE\log\det(\Z_{0}+\F_1^*\F_1),
\label{KeyII}
\end{align}
where we used a similar computation as in~\eqref{KeyI} for the fourth equality, and Theorem~\ref{th-Markov}-(b) is proven.

To establish Theorem~\ref{th-Markov}-(c), we follow the same strategy as in
the proof of Theorem~\ref{th-main}-(c):  Since the Markov chain
$(\F_n,\G_n,\Z_n)_{n\in\ZZ}$ is ergodic, we have 
\begin{equation} 
\label{kapoupaskap}
\kappa_\infty= \frac 1N \EE \ld(\Z_0+\F_1^* \F_1)
 = \lim_{n\to\infty}
    \frac{1}{Nn} \sum_{\ell=0}^{n-1}\ld(\Z_\ell+\F_{\ell+1}^*\F_{\ell+1})
  \quad\text{a.s.}
\end{equation} 
By using the same line of argument as in the proof of Lemma~\ref{le:uniqueQ0},
we obtain with a bound similar to~\eqref{boundDR0} and the arguments below
that $ \dr(\X_{n},\Z_{n}) \to 0$ in probability.  This implies in turn that
$\dr(\X_n + \F_{n+1}^* \F_{n+1},\Z_n + \F_{n+1}^* \F_{n+1})\leq 
\dr(\X_n, \Z_{n})\to0$, and thus, that 
$\ld(\X_n + \F_{n+1}^* \F_{n+1}) - \ld(\Z_{n} + \F_{n+1}^* \F_{n+1}) \to 0$ in 
probability. As a consequence, part (c) is obtained by taking a Ces\`aro 
average and~\eqref{kapoupaskap}. 
\end{proof}

\subsection{Proofs for Section~\ref{Examples}}
\label{proofExamples}

We shall need the following result, which follows from the fact that the zero
set of a non-zero polynomial of $d$ variables has zero measure for the Lebesgue
measure of $\RR^d$.

\begin{lemma}
\label{dens->fr} 
Let $X$ be a random complex $n\times n$ matrix whose distribution is absolutely continuous with respect to the Lebesgue 
measure on $\CC^{n\times n}\simeq \RR^{2n^2}.$ Then, $\PP( \rank(X) = n )= 1$. 
\end{lemma} 

We also need in this paragraph the following notations: Given a positive
integer $n$, we set $[n] :=\{0,\ldots, n-1\}$.  Given a matrix $X \in
\CC^{m\times n}$ and two sets of indices $J_1 \subset [m]$ and $J_2 \in [n]$,
we denote by $X^{J_1,J_2}$ the $|J_1|\times |J_2|$ submatrix of $X$ obtained by
keeping the rows of $X$ whose indices belong to $J_1$ and the columns of $X$
whose indices belong to $J_2$. We also write for convenience $X^{J_1, \cdot} :=
X^{J_1, [n]}$ and $X^{\cdot, J_2}: = X^{[m], J_2}$.  
Finally, we write $\log^-(x) = \min(\log x, 0)$ and 
$\log^+(x) = \max(\log x, 0)$.

\begin{proof}[Proof of Proposition~\ref{prop:com1}] 
We start with Assumption~\ref{Ass2}-(d). Using that $(U_n,V_n)$ and 
$(\F_{k},\G_k)_{k\leq n-1}$ are independent, it is enough to show that for any
$B,D\in  \CC^{N\times K}$, 
\begin{align}
\label{rangAR} 
& \PP \left[ \det( (V_n+D)^* (U_n+B) ) = 0 \right] = 0, \\
\label{rangAR2} 
&\forall v \in \CC^K \setminus \{ 0 \}, \qquad
\PP \left[ \Pi_{V_n+D}^\perp (U_n+B) v = 0 \right] = 0 .
\end{align} 

Letting $J:= [K]$ and $J^c := [N]\setminus [K]$, we have 
\begin{align} 
& \PP \left[ \det( (V_n+D)^* (U_n+B) ) = 0 \right] \nonumber \\
&= \PP \Bigl[ 
 \det( (V_n^{J,\cdot}+D^{J,\cdot})^* ( U_n^{J,\cdot} + B^{J,\cdot}) + 
    (V_n^{J^c,\cdot}+D^{J^c,\cdot})^* ( U_n^{J^c,\cdot} + B^{J^c,\cdot}) ) = 0 
 \Bigr]  . 
\label{PUV} 
\end{align} 
Since $U_n$ has a density (for Lebesgue), then for any invertible matrix $S \in \CC^{K\times K}$, we see that $S(U_n^{J,\cdot} + B^{J,\cdot})$ has a density.  Since
Lemma~\ref{dens->fr} yields that the random matrix $(V_n^{J,\cdot} + D^{J,\cdot})$ is
invertible a.s (it has a density),   the square matrix
$(V_n^{J,\cdot}+D^{J,\cdot})^* ( U_n^{J,\cdot} + B^{J,\cdot})$ has a density.
Recall that the convolution between an absolutely continuous probability and any
probability measure is absolutely continuous. Thus, since $(U_n^{J,\cdot},
V_n^{J,\cdot})$ and $(U_n^{J^c,\cdot}, V_n^{J^c,\cdot})$ are independent, the
matrix within the determinant at the right hand side of~\eqref{PUV} has a
density. Using Lemma~\ref{dens->fr} again, we obtain~\eqref{rangAR}. 

For any $v\in\CC^K \setminus \{ 0 \}$, the vector $w := (U_n+B) v$ is a 
random vector whose elements are independent and have probability densities.  
It results that for any matrix $C\in\CC^{N\times K}$,  we have
$\Pi_{C}^\perp w \neq 0$ a.s. Thus, 
$\PP \left[ \Pi_{V_n+D}^\perp (U_n+B) v = 0 \right] = 0$ by the Fubini-Tonelli
theorem, and~\eqref{rangAR2} is obtained. 

We now establish the truth of Assumption~\ref{Ass2}-(c). Write 
$\F_n = \begin{bmatrix} \f^0_n \ \cdots \ \f^{K-1}_n \end{bmatrix}$, where
$\f^k_n$ is the $k^{\text{th}}$ column of the matrix $\F_n$. For $k \in [K-1]$,
let $J_k = \{ k+1, \ldots, K-1 \}$. 
Applying, \emph{e.g.}, a Gram-Schmidt process to the successive columns 
$\f^0_n, \ldots, \f^{K-1}_n$, setting $\F^{\cdot, \emptyset}_n = 0 \in \CC^N$, 
and using the obvious inequality $\log^+ x \leq x$ for $x > 0$, we get that 
\begin{align}
\EE \left| \log \det \F_n^* \F_n \right| &= 
  \EE\left| \sum_{k=0}^{K-1} 
  \log (\f_n^k)^* \Pi_{\F^{\cdot, J_k}_n}^\perp \f_n^k \right| 
  \leq \EE \sum_{k=0}^{K-1} 
   \left| \log (\f_n^k)^* \Pi_{\F^{\cdot, J_k}_n}^\perp \f_n^k \right| 
   \nonumber \\
 &\leq \sum_{k=0}^{K-1} 
  \EE \left| \log^-((\f_n^k)^* \Pi_{\F^{\cdot, J_k}_n}^\perp \f_n^k) \right| 
 + \sum_{k=0}^{K-1} \EE \| \f_n^k \|^2 \nonumber \\
 &\leq \sum_{k=0}^{K-1} 
  \EE \left| \log^-((\f_n^k)^* \Pi_{\F^{\cdot, J_k}_n}^\perp \f_n^k) \right| 
 + C, 
\end{align} 
where $C < \infty$ since Assumption~\ref{Ass2}-(b) is satisfied.  
Fix $k \in [K]$. In the remainder of the proof, ``conditional'' refers to
a conditionning on $(\F_{n-1}, \G_{n-1}, u^{k+1}_n, \ldots, u^{K-1}_n)$. All 
the bounds are constants that only depend on the bound on the densities of the 
elements of $U_n$.  

The vector $\f_n^k$ can be written as $\f^k_n = \dd^k_{n-1} + u^k_n$, where 
$\dd^k_{n-1}$ is $(\F_{n-1}, \G_{n-1})$-measurable, and where $u^k_n$ is the 
$k^{\text{th}}$ column of $U_n$. By the assumptions on $(U_n)$, the elements 
of $\f^k_n$ are conditionally independent and have bounded densities. 
If $k < K-1$, make a 
$(\F_{n-1}, \G_{n-1}, u^{k+1}_n, \ldots, u^{K-1}_n)$-measurable
choice of a unit-norm vector $\p^k$ which is orthogonal to the subspace 
$\colspan \F^{\cdot, J_k}_n$, otherwise, take $\p^k$ as an arbitrary constant 
unit-norm vector. Since $| \log^-(\cdot)|$ is a nonincreasing function, 
$| \log^-((\f_n^k)^* \Pi_{\F^{\cdot, J_k}_n}^\perp \f_n^k) | 
\leq \left| \log^-(|\ps{\p^k, \f_n^k}|^2) \right|$.  
Since $\p^k = \begin{bmatrix} \p^k_0, \ldots, \p^k_{N-1} \end{bmatrix}^\T$ has 
unit-norm, it has at least one element, say $\p^k_0$, such that 
$| \p^k_0 | \geq 1/ \sqrt{N}$. Writing 
$\f^k_n = \begin{bmatrix} \f^k_{n,0}, \ldots, \f^k_{n,N-1} \end{bmatrix}^\T$, 
we get that the conditional density of $\p^k_0 \f^k_{n,0}$ is bounded, and 
by doing a simple calculation involving density convolutions, we finally 
obtain that $\ps{\p^k, \f_n^k}$ has a bounded conditional density. 
Now, it is easy to see that if $X$ is a complex random variable with a density 
bounded by a constant $C$ then $\EE | \log^-(|X|^2) | \leq C\pi$. 
This shows that  
$\EE \left| \log^-((\f_n^k)^* \Pi_{\F^{\cdot, J_k}_n}^\perp \f_n^k) \right| 
 < \infty$ for each $k\in [K]$, which completes the proof. 
\end{proof}

To prove Proposition~\ref{prop:com2}, we first need the following lemma.
\begin{lemma}
\label{PiG0} 
Given any positive integers $m,n,r$ satisfying $r \leq n \leq m$, let 
$X$ be a ${m\times n}$ matrix with rank $n$, write
$X = \begin{bmatrix} Y^\T \ \widetilde Y^\T \end{bmatrix}^\T$ where  
$Y$ is a $r\times n$ matrix, and assume that $\rank(Y) = r$. Then 
$\Pi_X^{[r],[r]} = I$ \emph{iff} 
$\colspan(\widetilde Y) = \colspan(\widetilde Y A)$ for some matrix $A$ satisfying  $\colspan(A)=\ker Y$. 
\end{lemma} 
\begin{proof}
The formula $\Pi_X = X (X^* X)^{-1} X^*$ yields
$\Pi_X^{[r],[r]} = Y ( Y^* Y + \widetilde Y^* \widetilde Y )^{-1} Y^*$. Performing a singular value decomposition, 
\begin{equation} 
Y = U \begin{bmatrix} \Lambda & 0 \end{bmatrix} 
\begin{bmatrix} V_1^* \\ V_2^* \end{bmatrix} ,
\end{equation} 
with $\Lambda$ the diagonal  
$r\times r$ matrix of singular values and $V_2$ satisfying $\colspan(V_2)=\ker Y$,
and using Schur's complement formula \eqref{Schur2},  we obtain
\begin{align} 
\Pi_X^{[r],[r]} &= 
 U \begin{bmatrix} \Lambda & 0 \end{bmatrix} 
\left( \begin{bmatrix} \Lambda^2 \\ & 0 \end{bmatrix} 
  + \begin{bmatrix} V_1^* \\ V_2^* \end{bmatrix} \widetilde Y^*  
   \widetilde Y\begin{bmatrix} V_1 & V_2 \end{bmatrix} \right)^{-1} 
 \begin{bmatrix} \Lambda \\ 0 \end{bmatrix} U^* \nonumber \\
 &= 
 U \Lambda \left( \Lambda^2 + V_1^* \widetilde Y^* ( I - \widetilde Y V_2 
(V_2^* \widetilde Y^* \widetilde Y V_2)^{-1} V_2^* \widetilde Y^* ) 
\widetilde Y V_1 \right)^{-1} \Lambda U^* \nonumber \\
&= 
U \Lambda \left( \Lambda^2 + V_1^* \widetilde Y^* 
\Pi_{\widetilde Y V_2}^\perp \widetilde Y V_1 \right)^{-1} \Lambda U^* . 
\end{align} 
This expression shows that $\Pi_X^{[r],[r]} = I$ \emph{iff} 
$V_1^* \widetilde Y^* \Pi_{\widetilde Y V_2}^\perp \widetilde Y V_1 = 0$. We 
then have 
\begin{align} 
V_1^* \widetilde Y^* \Pi_{\widetilde Y V_2}^\perp \widetilde Y V_1 = 0 
&\Leftrightarrow \colspan(\widetilde Y V_1 V_1^* \widetilde Y^*) \subset 
\colspan (\widetilde Y V_2 V_2^* \widetilde Y^*) \nonumber \\
&\Leftrightarrow 
\colspan(\widetilde Y \widetilde Y^*) \subset 
\colspan (\widetilde Y V_2 V_2^* \widetilde Y^*) \nonumber \\
&\Leftrightarrow 
\colspan(\widetilde Y ) = 
\colspan (\widetilde Y V_2 ) ,  
\end{align} 
which is the required result. 
\end{proof} 

\begin{proof}[Proof of Proposition~\ref{prop:com2}] 

Let us prove that Assumption~\ref{Ass2}-(d) holds. The recursive equation \eqref{defCn} satisfied by $(C_n)_{n\in\ZZ}$ yields,
for any $\ell \in [L-1]$ and  $k \in [L]$,
\begin{align}
c_{nL+\ell, k} &= H_k c_{nL+\ell-1, k} + u_{nL+\ell,k}  \nonumber \\
 &= H_k^2 c_{nL+\ell-2, k} + H_k u_{nL+\ell-1, k} + u_{nL+\ell,k} \nonumber \\
 &= \cdots \nonumber \\
 &= H_k^{\ell+1} c_{nL-1,k} + \sum_{i=0}^\ell H_k^i u_{nL+\ell-i, k} 
\end{align} 
where  $U_n =: \begin{bmatrix} u_{n,0}^\T \ \cdots \ u_{n,L}^\T \end{bmatrix}^\T$, 
the $u_{n,\ell}$'s being $R\times T$ matrices. Notice that the $c_{nL-1, k}$ and the $u_{nL+\ell-i, k}$ terms in the
rightmost term above are respectively $(\F_{n-1}, \G_{n-1})$-measurable and
independent from $(\F_{n-1}, \G_{n-1})$.  Plugging these equations in the
expressions for $\F_{n}$ and $\G_{n}$, we obtain 
\begin{align} 
\F_{n} &= \begin{bmatrix} 
u_{nL, L} & \cdots  & \cdots & u_{nL, 1} \\    
    & u_{nL+1,L} + H_L u_{nL,L}  & & \vdots   \\ 
    &   & \ddots   & \vdots                    \\ 
    &   &   &  u_{nL+L-1,L} + \sum_{i=1}^{L-1} H_L^i u_{nL+L-1-i,L} 
\end{bmatrix}  + B_{n-1} \nonumber \\
&= : Q_{n} + B_{n-1}, 
\label{struct-F} 
\end{align} 
and 
\begin{align} 
\G_{n} &= \begin{bmatrix} 
u_{nL, 0}        \\ 
 \vdots & u_{nL+1, 0} + H_0 u_{nL, 0} \\
  &   &  \ddots  \\ 
  & \cdots   &  \cdots  &  
     u_{nL+L-1, 0} +\sum_{i=1}^{L-1} H_0^i u_{nL+L-1-i, 0} 
\end{bmatrix} 
 + D_{n-1}  \nonumber \\
&=: S_{n} + D_{n-1}, 
\end{align} 
where the matrices $B_{n-1}$ and $D_{n-1}$ are $(\F_{n-1}, \G_{n-1})$-measurable
random matrices which are block-upper triangular and block-lower triangular
respectively, with $R\times T$ blocks (the exact expressions of these matrices
are irrelevant). Furthermore, the matrices $Q_n$ and $S_n$ are independent 
of $(\F_{n-1}, \G_{n-1})$. Thus, the proposition will be proven if we show that 
for all constant block-upper triangular matrices $B \in \CC^{LR\times LT}$ and 
all constant block-lower triangular matrices $D \in \CC^{LR\times LT}$ with
$R\times T$ blocks, 
\begin{align} 
\label{rangmultipath} 
&\PP \left[ \det( (S_{n}+D)^* (Q_{n}+B) ) = 0 \right] = 0,  \\
\label{rangmultipath2} 
&\forall v \in \CC^{LT} \setminus \{ 0 \}, \qquad
\PP \left[ \Pi_{S_{n}+D}^\perp (Q_{n}+B) v = 0 \right] = 0 .
\end{align} 
The matrix $(S_{n}+D)^* (Q_{n}+B)$ is a square $LT\times LT$ block-upper
triangular matrix with $T\times T$ blocks. Using Lemma~\ref{dens->fr} as in the
proof of Proposition~\ref{prop:com1},  one can check that all the diagonal
blocks of this matrix are a.s. invertible, and~\eqref{rangmultipath} is proven. 

To establish~\eqref{rangmultipath2}, we set $J_\ell := \{ \ell R, \ldots, \ell R  + R - 1\}$ and prove that
\begin{equation}
\label{projneq0} 
\forall \ell \in [L], \qquad
(\Pi_{S_{n}+D}^\perp)^{\cdot, J_\ell} \neq 0 \quad 
\text{a.s.} 
\end{equation} 
Indeed, given 
$v = [ v_0^\T, \ldots, v_{L-1}^\T ]^\T \in \CC^{LT} \setminus \{ 0 \}$ with
$v_i \in \CC^{T}$, let $k := \max \{ i \in [L] \, : \, v_i \neq 0 \}$. An inspection of~\eqref{struct-F} reveals that 
\begin{equation} 
(Q_{n} + B) v = \begin{bmatrix}  \vdots \\ 0 \\ 
 u_{nL+k,L} v_k \\ 0 \\ \vdots \end{bmatrix} + a ,  
\end{equation}  
for a random vector $a$ which is independent from $u_{nL+k,L}$. With this at
hand,  we see that
\begin{equation} 
\Pi_{S_{n}+D}^\perp (Q_{n}+B) v = 
   (\Pi_{S_{n}+D}^\perp)^{\cdot, J_k} u_{nL+k,L} v_k + 
 \Pi_{S_{n}+D}^\perp a \, .  
\end{equation} 
Since $\Pi_{S_{n}+D}^\perp$ and $u_{nL+k,L}$ are independent and 
$u_{nL+k,L} v_k$ has a density, \eqref{rangmultipath2} 
follows from~\eqref{projneq0}.

To complete the proof of that Assumption~\ref{Ass2}-(d) holds true, we now turn to the proof of  \eqref{projneq0}. We use the equivalence
$(\Pi_{S_{n}+D}^\perp)^{\cdot, J_\ell} = 0 \Leftrightarrow 
(\Pi_{S_{n}+D})^{J_\ell, J_\ell} = I$. Let us write 
\begin{equation} 
S_{n} + D = \begin{bmatrix} \ \widetilde Y_1 \ \\ \ Y \ \\ \ \widetilde Y_2 \ 
  \end{bmatrix} ,
\end{equation} 
where $Y = (S_{n} + D)^{J_\ell,\cdot} \in \CC^{R\times LT}$, and  set 
\begin{equation} 
\widetilde Y := 
 \begin{bmatrix} \ \widetilde Y_1 \ \\ \ \widetilde Y_2 \ \end{bmatrix} 
 \in \CC^{(L-1)R \times LT} . 
\end{equation} 
Since $\rank((\Pi_{S_n+D})^{J_\ell, J_\ell}) \leq \rank(Y)$, then if  
$\rank(Y) < R$ we have $(\Pi_{S_n+D})^{J_\ell, J_\ell} \neq I$. Assume 
$\rank(Y) = R$. Then $\dim \ker(Y) = LT - R$. By Lemma~\ref{PiG0}, 
$(\Pi_{S_n+D})^{J_\ell, J_\ell} = I$ implies 
$\rank \widetilde Y = \dim (\widetilde Y (\ker Y))$. Observe that 
$\dim (\widetilde Y (\ker Y)) \leq LT - R$.  For $m \in [L]$, let
$J'_m := \{ mR, \ldots, mR + T -1\}$ and  
$\widetilde J_\ell := \cup_{m \in [L]\setminus \{\ell\}} J'_m$. Then, 
$(S_n + D)^{\widetilde J_\ell, \cdot}$ is a submatrix of $\widetilde Y$.  
But thanks to the block-triangular stucture of $S_n + D$, one can check that 
$(S_n + D)^{\widetilde J_\ell, \cdot}$ has a block-echelon form, and its 
diagonal blocks $\{(S_n + D)^{J'_m,J'_m} \}_{m\neq \ell}$ are all a.s. invertible.  Thus, $\rank (S_n + D)^{\widetilde J_\ell, \cdot} = (L-1)T$ 
a.s.  Consequently, $\rank(\widetilde Y) \geq (L-1) T > LT - R 
 \geq \dim(\widetilde Y(\ker Y))$ a.s. which shows
that~\eqref{projneq0} holds true, and therefore, that 
Assumption~\ref{Ass2}-(d) is verified. 

We now turn to Assumption~\ref{Ass2}-(c). Getting back to 
Equation~\eqref{struct-F}, write 
\begin{equation} 
B_{n-1} = \begin{bmatrix} B_{n-1,0} & \times & \times \\ & \ddots &\times \\
  0 & & B_{n-1,L-1} \end{bmatrix} , 
\end{equation} 
where the $B_{n-1,\ell}$ are the $R\times T$ diagonal blocks of $B_{n-1}$. 
Defining $J := \{ 0, \ldots, T-1 \} \cup \{ R, \ldots, R+T-1 \} \cup \cdots 
 \cup \{ (L-1)R, \ldots, (L-1)R + T-1 \}$, we observe from 
Equation~\eqref{struct-F} that 
\begin{equation} 
\F_n^{J,\cdot} = 
 \begin{bmatrix} 
u_{nL, L}^{[T],\cdot}  + B_{n-1,0}^{[T],\cdot} & \times & \times \\    
 & \ddots & \times \\
 0 & & 
     u_{nL+L-1,L}^{[T],\cdot} + 
  ( B_{n-1,L-1} + \sum_{i=1}^{L-1} H_L^i u_{nL+L-1-i,L} )^{[T],\cdot} 
\end{bmatrix} 
\end{equation}  
is a square upper block-triangular matrix with $T\times T$ blocks. Moreover, 
the $\ell^{\text{th}}$ diagonal block of this matrix is the sum of 
$u_{nL+\ell,L}^{[T],\cdot}$ and a 
$(\F_{n-1}, \G_{n-1}, u_{nL}, \ldots, u_{nL+\ell-1})$-measurable term that we
denote by $\dd_{n,\ell}$. Now, since 
\begin{equation} 
(1 + \| \F_n \|^2) I 
> \F_n^* \F_n \geq 
  (\F_n^{J,\cdot})^*\F_n^{J,\cdot}
\end{equation} 
in the Hermitian semidefinite ordering, it holds that 
\begin{equation} 
LT \log (1 + \| \F_n \|^2) > \log\det(\F_n^* \F_n) \geq 
 \log\det( (\F_n^{J,\cdot})^*\F_n^{J,\cdot}), 
\end{equation} 
thus, 
\begin{equation} 
\EE |\log\det(\F_n^* \F_n)| < 
 \EE | \log\det( (\F_n^{J,\cdot})^*\F_n^{J,\cdot}) | + LT \,\EE \| \F_n \|^2 
 \leq 
 \EE | \log\det( (\F_n^{J,\cdot})^*\F_n^{J,\cdot}) | + C, 
\end{equation} 
where $C < \infty$ since Assumption~\ref{Ass2}-(b) is verified. Moreover, 
\begin{align} 
\EE | \log \det ((\F_n^{J,\cdot})^*\F_n^{J,\cdot}) | 
 &= \EE \Bigl| \sum_{\ell=0}^{L-1} 
  \log \det (u_{nL+\ell,L}^{[T],\cdot} + \dd_{n,\ell}) 
    (u_{nL+\ell,L}^{[T],\cdot} + \dd_{n,\ell})^* \Bigr|  \nonumber \\
 &\leq \sum_{\ell=0}^{L-1} 
  \EE \Bigl| \log \det (u_{nL+\ell,L}^{[T],\cdot} + \dd_{n,\ell}) 
    (u_{nL+\ell,L}^{[T],\cdot} + \dd_{n,\ell})^* \Bigr| , 
\end{align} 
and the summands in this last expression can be dealt with as in the
last part of the proof of Proposition~\ref{prop:com1}. 
The main distinctive feature of the proof here is that when we deal with the 
$\ell^{\text{th}}$ summand and when it comes to manipulate the conditional 
densities, we need to condition on 
$(\F_{n-1}, \G_{n-1}, u_{nL}, \ldots, u_{nL+\ell-1})$. This concludes the 
proof of Proposition~\ref{prop:com2}. 
\end{proof}

\section{Proof of Proposition~\ref{prop:mp}} 
\label{sec:RMT}

The expression of Shannon's mutual information per component given by
Theorem~\ref{th-main} provides a means of recovering the large random matrix
regime when $K,N\to\infty$ with $K/N\to\gamma\in(0,\infty)$ in a general
setting. We present a general result, then we particularize it to the setting
of Proposition~\ref{prop:mp}: 
\begin{lemma}
\label{le:RMT}  Under Assumption~\ref{Ass1}, if we introduce  for any $m\leq n$,
\begin{equation} 
\mathring \H_{m,n} := 
 \begin{bmatrix} \G_{m}   &          &     &  \F_m \\ 
                 \F_{m+1} & \G_{m+1} \\
                           & \ddots & \ddots \\
                           &        & \F_n & \G_n \end{bmatrix} ,
\end{equation} 
then we have as $M\to\infty$,
\begin{equation} 
\label{IM} 
\mathcal I_\rho = \frac1{(M+1)N}\,
\EE \ld (I+\rho \,\mathring \H_{0,M} \mathring \H_{0,M}^* ) +\mathcal O( 1/M)
\end{equation} 
where $\mathcal O( 1/M)$ is uniform in $K,N$.
\end{lemma}

As an illustration, we now prove Proposition~\ref{prop:mp} as an easy consequence of this lemma and well known results from  random matrix theory. 

\begin{proof}[Proof of Proposition~\ref{prop:mp}]

Observe from~\eqref{FGmultipath} and the assumptions made on the process
$(C_n)_{n\in\ZZ}$ that, for any $M\geq 1$, the $(M+1)N\times (M+1)N$ matrix
$\mathring H_{0,M}$  is a square matrix having independent entries with a
\emph{doubly stochastic} variance profile, and that the maximum of these
variances for a given $N$ is of order ${\mathcal O}(1/N)$. It is well known in
random matrix theory that when $N\to\infty$, the empirical spectral measure of
$\mathring H_{0,M}\mathring H_{0,M}^*$  converges narrowly to the
Marchenko-Pastur distribution $\mu_{\text{MP}}(d\lambda) = (2\pi)^{-1}
\sqrt{4/\lambda - 1} \1_{[0,4]}(\lambda) \, d\lambda$ a.s, see
\cite{gir-rand-determ90,TulVer04,hachem-loubaton-najim07}. Making a standard
moment control, we therefore obtain, for every fixed $M\geq 1$,
\begin{equation}  
\frac{1}{(M+1)N}\, 
 \EE \ld (I+\rho \,\mathring \H_{0,M} \mathring \H_{0,M}^* )  
 \xrightarrow[N\to\infty]{} \int \log(1 + \rho\lambda) \, 
   \mu_{\text{MP}}(d\lambda) . 
\end{equation} 
One can compute, see  \emph{e.g.} \cite[Th.~2.53]{TulVer04} or
\cite[Th.~4.1]{hachem-loubaton-najim07}, that this limiting integral coincides
with the right hand side of \eqref{LargeRMT}. Letting $M\to\infty$, the
proposition follows from Lemma~\ref{le:RMT}. 

\end{proof}

We finally turn to the proof of the lemma.
\begin{proof}[Proof of Lemma~\ref{le:RMT}] Using the notations of Theorem~\ref{th-main},  we set 
\begin{equation} 
\xi_n:=\log\det\left( I + \rho \,\F_n \W_{n-1} \F_n^* \right)- \log\det \W_n
\end{equation} 
and check, similarly as in  \eqref{relFinal}, that
\begin{equation} 
\xi_n=\log\det\left( I + \rho\, \G_n\,\G_n^*+\rho \,\F_n \W_{n-1} \F_n^* \right).
\end{equation} 
If we set  for convenience
\begin{align}
\begin{split} 
 \V_n & := \rho \,\G_n^* (I + \widetilde \V_n)^{-1} \G_n \\
 \widetilde\V_n & := \rho\, \F_n \W_{n-1} \F_n^*
\end{split} 
\end{align} 
then we have the relation $\widetilde \V_n = \rho\,\F_n ( I + \V_{n-1})^{-1} \F_n^*$ and  we moreover see that $\xi_n$ equals to

\begin{align}
&  \ld(I + \widetilde \V_n + \rho \,\G_n \G_n^* ) \nonumber \\
=& \ld(I + \rho \,\F_n (I + \V_{n-1})^{-1} \F_n^* + \rho \,\G_n \G_n^* ) \nonumber \\
=& \ld\left( 
I + \rho 
\begin{bmatrix} \F_n & \G_n \end{bmatrix}  
\begin{bmatrix} (I+\V_{n-1})^{-1} \\ & I \end{bmatrix} 
\begin{bmatrix} \F_n^* \\ \G_n^* \end{bmatrix}  
\right) \nonumber \\
=& \ld\left( 
I + \rho 
\begin{bmatrix} \F_n^* \\ \G_n^* \end{bmatrix} 
\begin{bmatrix} \F_n & \G_n \end{bmatrix} 
\begin{bmatrix} (I+\V_{n-1})^{-1} \\ & I \end{bmatrix} 
\right) \nonumber \\
=& \ld\left( 
I +  \begin{bmatrix} \V_{n-1} \\ & 0 \end{bmatrix} + 
\rho\begin{bmatrix} \F_n^* \\ \G_n^* \end{bmatrix} 
\begin{bmatrix} \F_n & \G_n \end{bmatrix} 
\right) - \ld(I+\V_{n-1}) \nonumber \\
=& \ld\left( 
I + \rho 
 \begin{bmatrix} \G_{n-1}^* (I+\widetilde \V_{n-1})^{-1/2} \\ 0 \end{bmatrix} 
 \begin{bmatrix} (I+\widetilde \V_{n-1})^{-1/2} \G_{n-1} & 0 \end{bmatrix} 
+ \rho 
\begin{bmatrix} \F_n^* \\ \G_n^* \end{bmatrix} 
\begin{bmatrix} \F_n & \G_n \end{bmatrix} 
\right) - \ld(I+\V_{n-1}) \nonumber \\
=& \ld\left( 
I + \rho 
 \begin{bmatrix} (I+\widetilde \V_{n-1})^{-1/2} \G_{n-1} & 0 \\ 
                \F_n & \G_n \end{bmatrix} 
 \begin{bmatrix} \G_{n-1}^* (I+\widetilde \V_{n-1})^{-1/2} & \F_n^* \\ 
  0 & \G_n^* \end{bmatrix} \right) - \ld(I+\V_{n-1}) \nonumber \\
=& \ld\left( 
I + \begin{bmatrix} \widetilde \V_{n-1} \\ & 0 \end{bmatrix} 
+ \rho 
 \begin{bmatrix} \G_{n-1} & 0 \\ 
                \F_n & \G_n \end{bmatrix} 
 \begin{bmatrix} \G_{n-1}^* & \F_n^* \\ 
  0 & \G_n^* \end{bmatrix} \right) 
 - \ld( I + \widetilde \V_{n-1}) - \ld(I+\V_{n-1}) .
\end{align}  
Using further the relation $I+\V_n=\W_n^{-1}$,
we thus obtain that 
\begin{equation} 
\xi_n + \xi_{n-1} = 
 \ld\left( 
I + \begin{bmatrix} \widetilde \V_{n-1} \\ & 0 \end{bmatrix} 
+ \rho 
 \begin{bmatrix} \G_{n-1} & 0 \\ 
                \F_n & \G_n \end{bmatrix} 
 \begin{bmatrix} \G_{n-1}^* & \F_n^* \\ 
  0 & \G_n^* \end{bmatrix} \right) . 
\end{equation} 
By iterating similar manipulations $M$ times, where $M$ will be made large in a moment, 
 we have
\begin{equation} 
\sum_{i=0}^M \xi_{n-i} = 
 \ld\left( I+\U_{n-M} 
 + \rho \,\hat \H_{n-M,n} \hat \H_{n-M,n}^* \right) \, , 
\end{equation} 
where we introduced
\begin{equation} 
\U_{m} := \begin{bmatrix} \widetilde \V_{m} \\ & 0 
  \end{bmatrix}  \quad \text{and} \quad 
 \hat\H_{m,n} := 
 \begin{bmatrix} \G_{m} &  \\ 
                 \F_{m+1} & \G_{m+1} \\
                           & \ddots & \ddots \\
                           &        & \F_n & \G_n \end{bmatrix} . 
\end{equation} 
By definition of $\xi_n$ and together with Theorem~\ref{th-main}(b), this yields the identity
\begin{equation} 
\mathcal I_\rho = \frac{1}{(M+1)N} \,
\EE \ld\left( I+
 \U_{0} + \rho\, \hat\H_{0,M} \hat \H_{0,M}^* \right) 
\end{equation} 
for all positive integers $M$.

Next, we control the cost of eliminating $\U_0$  from this expression. To do so, we  use that $|\log\det(I+A)|\leq\tr(A)$ and $\tr(AB)\leq \|B\|\tr(A)$ for any positive semi-definite Hermitian matrices $A,B$ and obtain
\begin{align}
|\ld\left( I+ \U_{0} 
   + \rho \,\hat\H_{0,M} \hat\H_{0,M}^* \right)-\ld\left( I 
   + \rho \,\hat\H_{0,M} \hat\H_{0,M}^* \right)|& \leq  \tr((I+\rho \,\hat\H_{0,M} \hat\H_{0,M}^*)^{-1}\U_0) \nonumber \\
&   \leq \tr(\U_0) \nonumber \\
& = \tr(\widetilde V_0)  \nonumber \\
&\leq \rho\,\tr(\F_0^*\F_0) \nonumber \\
&\leq \rho \min(K,N) \|F_0\|^2.
\end{align}
Using the moment assumption~\eqref{momFG}, this yields 
\begin{equation} \mathcal I_\rho = \frac1{(M+1)N}\,
\EE \ld (I+\rho \,\hat \H_{0,M} \hat \H_{0,M}^* ) +\mathcal O( 1/M)
\end{equation}  
where $\mathcal O( 1/M)$ is uniform in $K,N$. The same time of estimates yield that one can replace $\hat \H_{0,M} $ by $\mathring \H_{0,M}$ up to a $\mathcal O(1/M)$ correction, namely  
\begin{equation} 
\frac1{(M+1)N}\,
\EE \ld (I+\rho \,\hat \H_{0,M} \hat \H_{0,M}^* ) = \frac1{(M+1)N}\,
\EE \ld (I+\rho \,\mathring \H_{0,M} \mathring \H_{0,M}^* ) +\mathcal O( 1/M)
\end{equation} 
with $\mathcal O(1/M)$ uniform in $K,N$, and the lemma is proven.
\end{proof} 

\section{Conclusion} 
\label{conclusion} 

Shannon's mutual information of an ergodic wireless channel has been studied in
this paper under the weakest assumptions on the channel. The general capacity
result has been used to perform high SNR and the high dimensional analyses. 

Future research directions along the lines of this paper include the high SNR
analysis when the number of components at the receiver and at the transmitter
are equal. This analysis requires different tools than the ones used in
Section~\ref{sec:prfasymp} of this paper, which rely heavily on
Assumption~\ref{Ass2}--(d).  Another research direction is to thoroughly
quantify the impact of the parameters of a given statistical channel model on
the mutual information obtained by Theorems~\ref{th-main} and~\ref{th-Markov}.
In this respect, an attention can be devoted to the Doppler shift as in the
recent paper~\cite{gau-kob-bis-cai-arxiv19} and in the references therein.
Finally, transmission schemes with a partial channel knowledge at the receiver,
or scenarios with different delay constraints deserve a particular attention. 

\paragraph{Acknowledgements.}  The work of W.~Hachem was partially supported by
the French Agence Nationale de la Recherche (ANR) grant HIDITSA
(ANR-17-CE40-0003).  The work of A.~Hardy was partially supported by the Labex
CEMPI (ANR-11-LABX-0007-01) and the ANR grant BoB (ANR-16-CE23-0003).
S.~Shamai has been supported by the European Union's Horizon 2020 Research And
Innovation Programme, grant agreement no. 694630.

\def\cprime{$'$} \def\cdprime{$''$} \def\cprime{$'$} \def\cprime{$'$}
  \def\cprime{$'$} \def\cprime{$'$}


\begin{thebibliography}{10}

\bibitem{akh-glaz-93}
N.~I. Akhiezer and I.~M. Glazman.
\newblock {\em Theory of linear operators in {H}ilbert space}.
\newblock Dover Publications, Inc., New York, 1993.
\newblock Translated from the Russian and with a preface by Merlynd Nestell,
  Reprint of the 1961 and 1963 translations, Two volumes bound as one.

\bibitem{bad-bea-05}
K.~E. {Baddour} and N.~C. {Beaulieu}.
\newblock Autoregressive modeling for fading channel simulation.
\newblock {\em IEEE Transactions on Wireless Communications}, 4(4):1650--1662,
  July 2005.

\bibitem{bol-duh-hle-elssp03}
H.~B{\"o}lcskei, P.~Duhamel, and R.~Hleiss.
\newblock Orthogonalization of {OFDM/OQAM} pulse shaping filters using the
  discrete {Z}ak transform.
\newblock {\em Signal Processing}, 83(7):1379 -- 1391, 2003.

\bibitem{bol-hla-sp97}
H.~B{\"o}lcskei and F.~Hlawatsch.
\newblock Discrete {Z}ak transforms, polyphase transforms, and applications.
\newblock {\em IEEE Trans. Signal Processing}, 45(4):851--866, April 1997.

\bibitem{bou-93}
Ph. Bougerol.
\newblock Kalman filtering with random coefficients and contractions.
\newblock {\em SIAM J. Control Optim.}, 31(4):942--959, 1993.

\bibitem{cai-etal-comsurv18}
Y.~Cai, Z.~Qin, F.~Cui, G.~Y. Li, and J.~A. McCann.
\newblock Modulation and multiple access for {5G} networks.
\newblock {\em IEEE Com. Surveys \& Tutorials}, 20(1):629--646, Firstquarter
  2018.

\bibitem{car-lac-90}
R.~Carmona and J.~Lacroix.
\newblock {\em Spectral theory of random {S}chr\"odinger operators}.
\newblock Probability and its Applications. Birkh\"auser Boston, Inc., Boston,
  MA, 1990.

\bibitem{gau-kob-bis-cai-arxiv19}
L.~Gaudio, M.~Kobayashi, B.~Bissinger, and G.~Caire.
\newblock Performance analysis of joint radar and communication using {OFDM}
  and {OTFS}.
\newblock {\em CoRR}, abs/1902.01184, 2019.

\bibitem{gir-rand-determ90}
V.~L. Girko.
\newblock {\em Theory of random determinants}, volume~45 of {\em Mathematics
  and its Applications (Soviet Series)}.
\newblock Kluwer Academic Publishers Group, Dordrecht, 1990.
\newblock Translated from the Russian.

\bibitem{gray-entropy11}
R.~M. Gray.
\newblock {\em Entropy and information theory}.
\newblock Springer, New York, second edition, 2011.

\bibitem{hachem-loubaton-najim07}
W.~Hachem, Ph. Loubaton, and J.~Najim.
\newblock Deterministic equivalents for certain functionals of large random
  matrices.
\newblock {\em Ann. Appl. Probab.}, 17(3):875--930, 2007.

\bibitem{hmp-jmp15}
W.~Hachem, A.~M. Moustakas, and L.~Pastur.
\newblock The {S}hannon's mutual information of a multiple antenna time and
  frequency dependent channel: {A}n ergodic operator approach.
\newblock {\em J. Math. Phys.}, 56(11):113501, 29, 2015.

\bibitem{had-etal-wcnc17}
R.~Hadani, S.~Rakib, M.~Tsatsanis, A.~Monk, A.~J. Goldsmith, A.~F. Molisch, and
  R.~Calderbank.
\newblock Orthogonal time frequency space modulation.
\newblock In {\em 2017 IEEE WCNC}, pages 1--6, March 2017.

\bibitem{han-whi-93}
S.~V. Hanly and P.~Whiting.
\newblock Information-theoretic capacity of multi-receiver networks.
\newblock {\em Telecommunication Systems}, 1(1):1--42, Dec 1993.

\bibitem{kai-(livre)80}
T.~Kailath.
\newblock {\em Linear systems}.
\newblock Prentice-Hall, Inc., Englewood Cliffs, N.J., 1980.
\newblock Prentice-Hall Information and System Sciences Series.

\bibitem{kat-livre80}
T.~Kato.
\newblock {\em Perturbation theory for linear operators}.
\newblock Springer-Verlag, Berlin, second edition, 1976.
\newblock Grundlehren der Mathematischen Wissenschaften, Band 132.

\bibitem{kho-pastur93}
A.~M. Khorunzhy and L.~A. Pastur.
\newblock Limits of infinite interaction radius, dimensionality and the number
  of components for random operators with off-diagonal randomness.
\newblock {\em Comm. Math. Phys.}, 153(3):605--646, 1993.

\bibitem{lev-som-sha-zei-09}
N.~Levy, O.~Somekh, S.~Shamai, and O.~Zeitouni.
\newblock On certain large random {H}ermitian {J}acobi matrices with
  applications to wireless communications.
\newblock {\em IEEE Trans. Inform. Theory}, 55(4):1534--1554, 2009.

\bibitem{lev-zei-sha-it10}
N.~Levy, O.~Zeitouni, and S.~Shamai.
\newblock On information rates of the fading {W}yner cellular model via the
  {T}houless formula for the strip.
\newblock {\em IEEE Trans. Inform. Theory}, 56(11):5495--5514, 2010.

\bibitem{loz-tul-ver-it05}
A.~Lozano, A.~M. Tulino, and S.~Verd\'u.
\newblock High-{SNR} power offset in multiantenna communication.
\newblock {\em IEEE Trans. Inform. Theory}, 51(12):4134--4151, 2005.

\bibitem{maa-(livre)71}
H.~Maass.
\newblock {\em Siegel's modular forms and {D}irichlet series}.
\newblock Lecture Notes in Mathematics, Vol. 216. Springer-Verlag, Berlin-New
  York, 1971.

\bibitem{mar-etal-(livre)16}
T.~L. Marzetta, E.~G. Larsson, H.~Yang, and H.~Q. Ngo.
\newblock {\em Fundamentals of massive {MIMO}}.
\newblock Cambridge University Press, 2016.

\bibitem{mey-twe-livre09}
S.~Meyn and R.L. Tweedie.
\newblock {\em Markov Chains and Stochastic Stability}.
\newblock Cambridge Mathematical Library. Cambridge University Press, 2009.

\bibitem{mul-it02}
R.~R. M\"{u}ller.
\newblock A random matrix model of communication via antenna arrays.
\newblock {\em IEEE Trans. Inform. Theory}, 48(9):2495--2506, 2002.

\bibitem{pas-fig-92}
L.~Pastur and A.~Figotin.
\newblock {\em Spectra of random and almost-periodic operators}, volume 297 of
  {\em Grundlehren der Mathematischen Wissenschaften [Fundamental Principles of
  Mathematical Sciences]}.
\newblock Springer-Verlag, Berlin, 1992.

\bibitem{pol-sze-(livre)98}
G.~P\'{o}lya and G.~Szeg{\"{o}}.
\newblock {\em Problems and theorems in analysis. {I}}.
\newblock Classics in Mathematics. Springer-Verlag, Berlin, 1998.
\newblock Series, integral calculus, theory of functions, Translated from the
  German by Dorothee Aeppli, Reprint of the 1978 English translation.

\bibitem{ree-sim-1}
M.~Reed and B.~Simon.
\newblock {\em Methods of modern mathematical physics. {I}}.
\newblock Academic Press Inc., New York, second edition, 1980.
\newblock Functional analysis.

\bibitem{TulVer04}
A.~Tulino and S.~Verd{\'u}.
\newblock Random matrix theory and wireless communications.
\newblock In {\em Foundations and Trends in Communications and Information
  Theory}, volume~1, pages 1--182. Now Publishers, June 2004.

\bibitem{tul-cai-sha-ver-10}
A.~M. Tulino, G.~Caire, S.~Shamai, and S.~Verd\'{u}.
\newblock Capacity of channels with frequency-selective and time-selective
  fading.
\newblock {\em IEEE Trans. Inform. Theory}, 56(3):1187--1215, 2010.

\bibitem{wyn-it94}
A.~D. Wyner.
\newblock Shannon-theoretic approach to a {G}aussian cellular multiple-access
  channel.
\newblock {\em IEEE Trans. Inform. Theory}, 40(6):1713--1727, Nov 1994.

\end{thebibliography}
\end{document}